\documentclass[preprint,pre,nobibnotes,twocolumn,10pt]{revtex4}%
\usepackage{amsfonts}
\usepackage{amsmath}
\usepackage{amssymb}
\usepackage{graphicx}%
\providecommand{\U}[1]{\protect\rule{.1in}{.1in}}
\newtheorem{theorem}{Theorem}

\newenvironment{proof}[1][Proof]{\noindent\textbf{#1.} }{\ \rule{0.5em}{0.5em}}
\begin{document}
\preprint{ }
\preprint{UATP/1904}
\title{A First-Principles Nonequilibrium Deterministic Equation of Motion of a
Brownian Particle and Microscopic Viscous Drag }
\author{P.D. Gujrati,$^{1,2}$ }
\affiliation{$^{1}$Department of Physics, $^{2}$Department of Polymer Science, The
University of Akron, Akron, OH 44325}
\email{pdg@uakron.edu}

\begin{abstract}
We present a \emph{first-principles thermodynamic approach} to provide an
alternative to the Langevin equation by identifying the \emph{deterministic}
(no stochastic component) microforce $\mathbf{F}_{k,\text{BP}}$ acting on a
nonequilibrium Brownian particle (BP) in its $k$th microstate $\mathfrak{m}%
_{k}$. (The prefix micro refers to microstate quantities and carry a suffix
$k$.) The deterministic new equation is easier to solve using basic calculus.
Being oblivious to the second law, $\mathbf{F}_{k,\text{BP}}$ does \emph{not}
always oppose motion but viscous dissipation emerges upon ensemble averaging.
The equipartition theorem is always satisfied. We reproduce well-known results
of the BP in equilibrium. We explain how the microforce is obtained directly
from the mutual potential energy of interaction beween the BP and the medium
after we average it over the medium so we only have to consider the particles
in the BP. Our approach goes beyond the phenomenological and equilibrium
approach of Langevin and unifies nonequilibrium viscous dissipation from
mesoscopic to macroscopic scales and provides new insight into Brownian motion
beyond Langevin's and Einstein's formulation.

\end{abstract}
\keywords{Nonequilibrium Brownian particle; viscous dissipation; Langevin equation;
microstates; internal microwork; random variables and fluctuations; white
Gaussian noise; irreversibility, internal equilibrium states.}\date{\today}
\maketitle

\section{Introduction}

The aim in this study is to introduce a nonequilibrium (NEQ) thermodynamics
based exclusively on microstates, which will be called the $\mu$NEQT in short
($\mu$ for micro-), and apply it to describe viscous dissipation associated
with the dynamics of a Brownian particle (BP) as it undergoes a
\emph{macroscopic} relative motion with respect to the rest of the system
$\Sigma$. The system is in a medium $\widetilde{\Sigma}$; see Fig.
\ref{Fig_System}. Due to the above motion, $\Sigma$ is not in equilibrium (EQ)
\cite{Landau,Note}; however, $\widetilde{\Sigma}$ is always assumed to be in
EQ. The $\mu$NEQT will be an extension of the traditional macroscopic NEQ
thermodynamics (MNEQT, M for macro-)
\cite{deGroot,Prigogine,Coleman,Maugin,Gujrati-I,Gujrati-II,Gujrati-III} to
the microstate level.

At the simplest level, BP's diffusion and dynamics in EQ are described using
Einstein's and Langevin's approaches, respectively
\cite{Einstein-BrownianMotion,Langevin,Chandrasekhar,Sekimoto-Book}. The study
is motivated by the fact that the dynamics of a BP has received a resurgence
of interest mainly due to the current interest in nonequilibrium (NEQ)
processes observed at the microstate scale such as by micron- or smaller-sized
\emph{active }BPs often encountered in biological or man-made systems
\cite{Marconi,Romanczuk,Kapral-2017,Fodor}, and in inhomogeneous systems
\cite{Beck}. These processes are strongly influenced by NEQ fluctuations that
may be very different from their equilibrium counterpart.
\begin{figure}
[ptb]
\begin{center}
\includegraphics[
height=2.6238in,
width=3.2145in
]%
{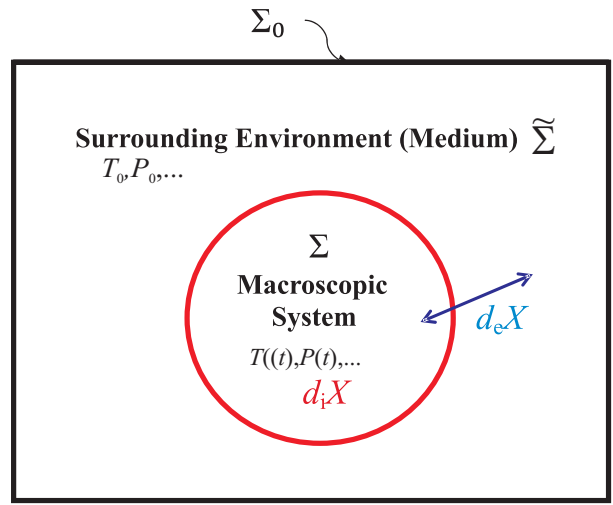}%
\caption{An isolated system $\Sigma_{0}$ consisting of the system $\Sigma$ in
a surrounding medium $\widetilde{\Sigma}$. The BP, which is not shown here, is
embedded within $\Sigma$ as shown in Fig. \ref{Fig_Piston-Spring}. The medium
and the system are characterized by their fields $T_{0},P_{0},...$ and
$T(t),P(t),...$, respectively, which are different when the two are out of
equilibrium. Exchange quantities ($d_{\text{e}}X$) carry a suffix "e" and
irreversibly generated quantities ($d_{\text{i}}X$) within the system by a
suffix "i" by extending the Prigogine notation. Their sum $d_{\text{e}%
}X+d_{\text{i}}X$ is denoted by $dX$, which is a system-intrinsic quantity
(see text). }%
\label{Fig_System}%
\end{center}
\end{figure}

Spontaneous fluctuations close to EQ are Gaussian \cite{Landau} as in the two
approaches above, but \emph{non-Gaussianity }\cite{Granick,Zheng} seems to be
a signature of NEQ states and abounds in Nature when the system is far from
equilibrium. In this case, the above two approaches must fail and we need to
develop new approaches to study NEQ viscous drag. Several attempts have been
made to obtain generalized Langevin equations for the microstate
$\mathfrak{m}_{k}$; see for example Ref. \cite{Evans}. Despite significant
attempts to understand Brownian dynamics in a passive or active medium under
external driving
\cite{Mizuno,Hanggi,Huang,Sekimoto-Book,Marconi,Kapral-2017,Gaspard} resulting
in NEQ conditions, we still lack its comprehensive thermodynamic
understanding, gaining which should then allow us to have a systematic
enlargement of the NEQ state space $\mathfrak{S}$ (see below) and expansion to
higher order than just two in fluctuations. It was Einstein
\cite{Einstein-BrownianMotion} who had first initiated a successful
\emph{thermodynamic} approach for a BP in EQ. This should be contrasted with
the mechanical \emph{stochastic} approach of Langevin \cite{Langevin}. We will
adopt a \emph{hybrid} approach in this work in which we begin with a NEQ
thermodynamic from which we derive a mechanical equation of motion. Being
associated with microstates, the $\mu$NEQT will allow us to capture the
thermodynamics of fluctuations and viscous dissipation experienced by a BP
under any condition using the state space $\mathfrak{S}$. Within the framework
of this theory, the behavior of the system will dictate whether fluctuations
are Gaussian or not or whether viscous drag follows Stokes' law or a more
complex behavior. Moreover, while most of us are familiar with classical
MNEQT, not many are trained in the technical issues of the Wiener process
(such as the It\^{o} and Stratonovich\ integrals) necessary to follow
Langevin's \emph{stochastic} approach. In our approach, we will only be
dealing with a \emph{deterministic} equation of motion. This should make our
approach quite useful.

Einstein assumed that a BP can be simply described by its stochastic
\emph{center of mass }(CM)\emph{ }position\emph{ }$\mathbf{r}_{k}$\emph{\ }for
its specification and by ignoring the center of mass momentum $\mathbf{p}_{k}%
$, and the specification of its constituent atoms or molecules that identify
the BP as a thermodynamic \emph{object}. The interface between the BP and the
system causes osmotic pressure that drives the diffusion of its CM. The EQ
diffusion of the BP obeys a diffusion equation, the Fokker--Planck equation
describing stochasticity in terms of conditional probabilities in the ensemble
picture \cite{Keizer}, which Einstein solved. Langevin \cite{Langevin} later
provided a \emph{stochastic} formulation of the same motion by applying
Newton's equation
\begin{equation}
Md^{2}\mathbf{r}_{k}/dt^{2}=\mathbf{F}_{k,\text{BP}}^{\prime}%
(t)\label{LangevinEquation0}%
\end{equation}
to the BP of mass $M$ in each \emph{microstate }$\mathfrak{m}_{k}$ specified
by a small cell around $(\mathbf{r}_{k}(t),\mathbf{v}_{k}(t))$, by dividing
the stochastic force $\mathbf{F}_{k,\text{BP}}^{\prime}$ into a
\emph{deterministic} (no randomness) force component $\mathbf{F}_{k,\text{f}%
}(t)\doteq-\gamma\mathbf{v}_{k}(t)=-\gamma d\mathbf{r}_{k}(t)/dt$ determined
by the microstate $k$ (with $\gamma>0$), and a \emph{stochastic} Gaussian
white force component $\boldsymbol{\xi}(t)$ \cite{Mazur}; see
\cite{Chandrasekhar} for an elegant discussion and inherent assumptions. Both
$\gamma$ and the Langevin force $\boldsymbol{\xi}(t)$ are \emph{independent}
of the position and velocity of the BP in $\mathfrak{m}_{k}$ so the two forces
are independent despite arising from the interaction of the BP with its
surroundings. Chandrasekhar \cite{Chandrasekhar} emphasizes $\xi(t)$ as a
\emph{characteristic} of a BP, which undergoes \emph{rapid} fluctuations over
an interval $\Delta t$ over which $\mathbf{v}_{k}(t)$ only undergoes a small
variation. Implicit in the above formulation is that (i) $\mathbf{F}%
_{k,\text{f}}(t)$ opposes motion in\emph{ every }$\mathfrak{m}_{k}$\emph{
}($\gamma>0$)\emph{ }as if it is a \emph{macroscopic}, \textit{i.e., }a
thermodynamic average force satisfying the second law, (ii) the Langevin force
$\boldsymbol{\xi}(t)$ performs no average work, (iii) $\boldsymbol{\xi}(t)$
represents a rapidly fluctuating force (fast-force) and $\mathbf{F}%
_{k,\text{f}}(t)$ a slowly varying force (slow-force) due to widely separated
time scales, and (iv) the separation between the two distinct time scales
requires two distinct averages involving a joint probability distribution of
initial microstates $(\mathbf{r}_{k}(0),\mathbf{v}_{k}(0))$ and
$\boldsymbol{\xi}(t)$; the latter requires its conditional probability
distribution corresponding to a Wiener process \cite{Keizer}. The
\emph{separation} between $\mathbf{F}_{k,\text{f}}(t)$\ and $\boldsymbol{\xi
}(t)$ is one of the basic assumptions as discussed by Chandrasekhar
\cite{Chandrasekhar}; see also Mazur and Bedeaux \cite{Mazur}, and Pomeau and
Piasecki \cite{Pomeau}. The above four assumptions are taken to be valid in
any theory of a BP in which a clear separation between fast and slow
components of the force are made. For brevity, we will call all of them as
following the \emph{Langevin approach}, which also includes the modern theory
of stochastic processes \cite{Keizer} and the Mori-Zwanzig approach
\cite{Evans-book,Zwanzig}.

The distinct approaches\ by Einstein and Langevin have developed into
mathematically distinct but physically equivalent ways to investigate
stochastic processes \cite{Keizer}. The approach by Einstein adopts a
probabilistic approach to capture thermodynamic stochasticity and results in
\emph{ensemble},\emph{ }\textit{i.e.,}\emph{ thermodynamic averages} such as
the root-mean-square displacement but dynamics is not a central issue. In
contrast, Langevin's approach starts with the dynamical equation in which
$\mathbf{F}_{k,\text{f}}(t)$ is related to the instantaneous velocity
$\mathbf{v}_{k}(t)$. The stochasticity due to $\boldsymbol{\xi}(t)$ defines a
stationary process because the probability distribution does not change in
time \cite{Chandrasekhar}, a well-known property of white noise so that
averaging Eq. (\ref{LangevinEquation0}) over $\boldsymbol{\xi}(t)$ alone
results in a deterministic equation.

All quantities associated with $\mathfrak{m}_{k}$ are called
\emph{microquantities} as opposed to their ensemble averages, which we call
\emph{macroquantities} or simply quantity in this work. All microquantities
will always carry a suffix $k$
\cite{Gujrati-GeneralizedWork,Gujrati-GeneralizedWork-Expanded}.

The Langevin equation in one dimension is
\begin{equation}
Mdv_{k}(t)/dt=F_{k,\text{BP}}^{\prime}(t)=-\gamma v_{k}(t)+\xi
(t),\label{LangevinEquation}%
\end{equation}
with $F_{k,\text{f}}(t)\doteq-\gamma v_{k}(t)$. The sets $\{v_{k}(t)\}$ and
$\{F_{k,\text{BP}}^{\prime}(t)\}$ form the set of outcomes of random variables
\textsf{v} and \textsf{F}$^{^{\prime}}$, respectively, over $\{\mathfrak{m}%
_{k}\}$.

Stokes' law for a spherical BP of radius $a$ gives $\gamma=6\pi a\eta>0$,
where $\eta$ is the viscosity of the surrounding fluid; see Ref.
\cite{Kapral0} for a microscopic derivation.

It is well known that the Langevin force $\xi(t)$ is central to satisfy the
equipartition theorem $\left\langle v^{2}(t)\right\rangle =T_{0}/M$, where
$\left\langle \bullet\right\rangle $ refers to the ensemble average over all
microstates and $T_{0}$ is the temperature of the medium (heat bath)
\cite{Langevin}. Therefore, one needs to perform two \emph{distinct} and
\emph{independent} averages over initial velocities $\left\{  v_{0}\right\}  $
and positions $\left\{  x_{0}\right\}  $, and $\xi(t)$ at each time; see for
example Reichl \cite{Reichl} for a clear discussion.\ On the other hand, the
equipartition theorem is always fulfilled in the Einstein approach
\cite{Einstein-BrownianMotion} without any $\xi$. This suggests that our
hybrid approach based on a statistical formulation ($\mu$NEQT) \`{a} la
Einstein, from which equations of motion \`{a} la Langevin can also be
derived, will offer a possible route to study all possible kinds of BPs since
it will contain all the information necessary to incorporate the
\emph{correct} (Gaussian or non-Gaussian) nature of fluctuations appropriate
for the system. No $\xi(t)$ is required. Thus, the $\mu$NEQT offers a
framework to study all of them within a unified first-principles approach in
which only the ensemble average $\left\langle \bullet\right\rangle $ is required.

Our approach using the $\mu$NEQT is very different from the above two
approaches. As various $\mathfrak{m}_{k}$'s are defined by the Hamiltonian
$\mathcal{H}$, we identify the microforce $\mathbf{F}_{k,\text{BP}}(t)$ on
$\mathfrak{m}_{k}$ as the \emph{mechanical} force determined by the microstate
energy (microenergy) $E_{k}$ obtained directly from $\mathcal{H}$. This
microforce refers to the system $\Sigma$ and not to the BP, unless the BP
happens to be the system as will be the case in Sec. \ref{Sec-Example}. As
$\mathcal{H}$ itself is deterministic, $\mathbf{F}_{k,\text{BP}}(t)$ is
\emph{deterministic}, which immediately distinguishes it from the stochastic
force $\mathbf{F}_{k,\text{BP}}^{\prime}(t)$ used by Langevin. In addition,
$\ \mathbf{F}_{k,\text{BP}}(t)$\emph{ }is\emph{ }not partitioned\emph{
}into\emph{ slow }and\emph{ fast }components as is required in the Langevin
approach.\emph{ }Newton's equation with the deterministic $\mathbf{F}%
_{k,\text{BP}}(t)$ is different from Eq. (\ref{LangevinEquation0}) and much
simpler to solve as we do not need to deal with stochastic integrals
\cite{Keizer}. This makes solving the equation of motion straight forward
using basic calculus. There is no requirement that $\mathbf{F}_{k,\text{BP}%
}(t)$ oppose the motion in $\mathfrak{m}_{k}$\ as the second law is applicable
to macrostates and not to microstates. Thus, it is distinct from the slow
component $\mathbf{F}_{k,\text{f}}(t)$\ above. As is normal, $\mathbf{F}%
_{k,\text{BP}}(t)$ fluctuates over $\left\{  \mathfrak{m}_{k}\right\}  $. Its
ensemble average $\mathbf{F}_{\text{BP}}(t)$ satisfies the second law and
opposes the motion, whereas $\mathbf{F}_{k,\text{f}}(t)$\ satisfies the law
($\gamma>0$) for each microstate. We calculate various fluctuations over the
statistical ensemble in the $\mu$NEQT and reproduce all known results. The
probabilities $\left\{  p_{k}(t)\right\}  $ are determined uniquely in the
$\mu$NEQT as we will see. We only focus on $\left\{  \mathbf{F}_{k,\text{BP}%
}(t)\right\}  $ and the consequences here. We describe in detail the
computational scheme to show the feasibility and the usefulness of our approach.

The layout of the paper is as follows. We introduce the new thermodynamics in
the next section with a focus on the BP problem and give a very general form
of viscous dissipation that follows from the second law. In Sec.
\ref{Sec-InternalContribution}, we discuss in depth the microforce that
results in the viscous dissipation, the resulting new microstate equation of
motion, and calculate various thermodynamic fluctuations. Sec.
\ref{Sec-Example} deals with the feasibility of the new approach for the
simple case of a BP in a medium. We consider here the mutual interaction
between the BP and the medium and average it over the macrostate of the
medium. The resulting potential depends only on the BP-microstate and
determines $\mathbf{F}_{k,\text{BP}}(t)$. Thus, we only need to pay attention
to the particles in the BP, which simplifies the calculation. The final
section deals with discussion and conclusions.

\section{A New Approach using the $\mu$NEQT\label{Sec-GeneralCase}}

We find it very useful to follow the extension of the Prigogine's notation in
this study \cite{Note-Notation}; see also Fig. \ref{Fig_System} caption.

\subsection{The Concept of Internal Equilibrium \label{Sec-IEQ}}

The central concept of the NEQT exploited here is that of the \emph{internal
equilibrium }(IEQ) according to which the entropy $S$ of a NEQ macrostate is a
\emph{state function} of the state variables in the enlarged state space
$\mathfrak{S}$ \cite{Gujrati-I,Gujrati-II,Gujrati-III}; see Sec.
\ref{Sec-MicroTH} for details. The enlargement of the space relative to the EQ
state space $\mathfrak{S}_{0}$ is due to independent internal variables
\cite{Coleman,deGroot,Prigogine,Maugin,Gujrati-I,Gujrati-II} that are required
to describe a NEQ macrostate as we explain below. In EQ, the internal
variables are no longer independent of the observables forming the space
$\mathfrak{S}_{0}$. As a consequence, their affinities vanish in EQ.
Observables are quantities that can be controlled from the outside but not the
internal variables. In general, the temperature $T$ of the system in IEQ is
identified in the standard manner by the relation
\begin{equation}
1/T=\partial S/\partial E \label{IEQ-Temp}%
\end{equation}
using the fact that $S$ is state variable in $\mathfrak{S}$.

An important property of IEQ macrostates is the following that will prove very
useful here: It is possible in an IEQ macrostate to have different degrees of
freedom or different parts of a system to have different temperatures than
$T$. For example, in a glass, it is well known that the vibrational degrees of
freedom have a different temperature than the configurational degrees of
freedom \cite{Debenedetti,Gujrati-Hierarchy}. In the viscous drag problem, the
CM-motion of the BP can be separated out from the motion of its various
constituent particles as is well known; see Sec. \ref{Sec-Example}. Then, it
is possible for the BP motion to have a different temperature than $T$
introduced above. This observation is easily verified in MNEQT based on the
concept of IEQ as done elsewhere \cite[see Sec. 8.1 and Eq. (58)]%
{Gujrati-Hierarchy}. The derivation also works when various parts of the
system have different temperatures. As this observation will play an important
role in this investigation, we rederive it for clarity in a different manner,
which supplements the previous demonstration \cite{Gujrati-Hierarchy} and also
shows how an internal variable is required to describe an IEQ macrostate.

\subsubsection{An Example}

Consider the case of two identical bodies $\Sigma_{1}$ and $\Sigma_{2}$ in
thermal contact at different temperatures $T_{1}(t)$ and $T_{2}(t)$ and
energies $E_{1}(t)$ and $E_{2}(t)$, respectively; we ignore other observables
$N,V$, etc. We assume that each one is in an EQ state of its own at each
instant. Together, they form an isolated system $\Sigma$, whose entropy
$S(E_{1},E_{2})=S_{1}(E_{1})+S_{2}(E_{2})$ is a function of two variables at
each instant $t$, and can be written as a state function in the enlarged state
space formed by $E=E_{1}+E_{2}=const$ (the observable) and $\xi(t)=E_{1}%
-E_{2}$ (the internal variable). (We have neglected the interaction energy
$E_{12}$\ between $\Sigma_{1}$ and $\Sigma_{2}$ here.) For this IEQ state, it
is trivial to show that the temperature is $T(t)=2T_{1}T_{2}/(T_{1}+T_{2})$
and the affinity $T\partial S/\partial\xi$ is $A(t)=(T_{1}-T_{2})/(T_{1}%
+T_{2})$. At equilibrium, $T_{1}=T_{2}=T_{\text{eq}}$ and $\xi=0,A=0$. Thus,
$T_{1}$ and $T_{2}$ may be very different, yet the system can be treated in
IEQ, any temperature difference between its parts not withstanding. The
discussion can be extended easily to the case the two bodies are in IEQs and
also when they are of different sizes.

\subsubsection{Microstates}

We consider the phase space $\Gamma$ associated with $\Sigma$ and partition it
completely into countable nonoverlapping cells $\left\{  \delta\mathbf{z}%
_{k}\right\}  ,k=1,2,\ldots$, each of size $h^{3N}$,\ around the phase point
$\mathbf{z\in\Gamma}$; here $N$ is the number of particles in $\Sigma$ and we
assume that the volume $\left\vert \Gamma\right\vert $ of $\Gamma$ has been
divided by $N!$ to account for the permutation symmetry of the $N$ particles.
We use the cells to identify the set of microstate $\left\{  m_{k}\right\}  $
of $\Sigma$. Consider $\Sigma$ to be composed of two distinct bodies
$\Sigma_{1}$ and $\Sigma_{2}$, as above in the example. As each cell
$\delta\mathbf{z}_{k}$ is a union of cells $\delta\mathbf{z}_{k_{1}}^{(1)}$
and $\delta\mathbf{z}_{k_{2}}^{(2)}$\ corresponding to $\Sigma_{1}$ and
$\Sigma_{2}$, we can relate the microstate energies as follows:
\begin{subequations}
\begin{equation}
E_{k}=E_{k_{1}}+E_{k_{2}}+E_{k,12} \label{MicroEnergy-Relation}%
\end{equation}
where we have also included the interaction energy $E_{k,12}$, which is
usually neglected as we did above. These energies are independent of the
macrostates and, therefore, independent of quantities such as the temperatures
that specify macrostates of various bodies forming the system. The energies
corresponding to their macrostates are related by%
\begin{equation}
E=E_{_{1}}+E_{_{2}}+E_{12}. \label{MacroEnergy-Relation}%
\end{equation}

The interaction energy $E_{k,12}$ and its macroaverage $E_{12}$, however, will
play an important role later in Sec. \ref{Sec-Example}, where we deal with
relative motion between $\Sigma$ (the BP) and $\widetilde{\Sigma}$; the
existence of this motion is central for viscous dissipation as we will see in
Sec. \ref{Sec-MicroTH}.

\subsection{Ensemble Stochasticity and the Second Law}

We consider a system $\Sigma$, see Fig. \ref{Fig_System}, that contains a
single BP that is shown explicitly in Fig. \ref{Fig_Piston-Spring} as part of
$\Sigma$. The single BP is our focus in this work. We follow the standard
formulation for a statistical system $\Sigma$ \cite{Landau}, which interacts
weakly with a much larger medium $\widetilde{\Sigma}$ so this interaction
$U_{\text{int}}$ is normally ignored. This is possible as we do not allow any
relative motion between $\Sigma$ and $\widetilde{\Sigma}$ as noted above; see
also Sec. \ref{Sec-MicroTH}. However, $U_{\text{int}}$ must not be zero
identically otherwise there cannot be any energy (heat and work) exchange
between $\Sigma$ and $\widetilde{\Sigma}$. Together, they form an isolated
system $\Sigma_{0}\doteq\Sigma\cup\widetilde{\Sigma}$. The system may be far
away from equilibrium so the new theory is more general than the EQ treatments
by Einstein and Langevin.

As said above, treating NEQ states normally requires some (extensive) internal
variables that are generated due to \emph{internal processes}
\cite{Coleman,deGroot,Prigogine,Maugin,Gujrati-I,Gujrati-II}. Their conjugate
fields, called \emph{affinity}, vanish only in equilibrium. The system is
specified by a Hamiltonian $\mathcal{H}(\left.  \mathbf{z}\right\vert
\mathbf{Z})$ in which $\mathbf{z}$ denotes a phase point in its phase space
and $\mathbf{Z}\doteq\left\{  Z\right\}  $ denotes the set of parameters such
as the volume $V$, the number of particles $N$ which we do not show, etc. and
internal variables.

The time dependence in some or all components in $\mathbf{Z}$ gives rise to
time dependence in the Hamiltonian $\mathcal{H}(\left.  \mathbf{z}\right\vert
\mathbf{Z})$; the dynamical variable $\mathbf{z}$ plays no role as we show in
Eqs. (\ref{HamiltonianChange}) and (\ref{HamiltonianChange-Work}). From
$\mathcal{H}(\left.  \mathbf{z}\right\vert \mathbf{Z})$,\ we identify
microstates $\mathfrak{m}_{k}(\mathbf{Z})$ and their microenergies
$E_{k}(\mathbf{Z})$; we will usually suppress the $\mathbf{Z}$-dependence
unless necessary for clarity. The microstate $\mathfrak{m}_{k}$ appears with
probability $p_{k}$\ in the statistical ensemble. The set $\left\{
p_{k}\right\}  $ determines the stochasticity in the ensemble. Accordingly, it
determines the nature of the macrostate (EQ vs NEQ) but the sets $\left\{
E_{k}\right\}  $ and $\left\{  \mathfrak{m}_{k}\right\}  $ are independent of
$\left\{  p_{k}\right\}  $ so they are deterministic.

In the $\mu$NEQT, the two aspects can be separated out in an unambiguous
fashion so we can uniquely determine the deterministic quantities such as
$\left\{  F_{k,\text{BP}}\right\}  $. Accordingly, we do not need to partition
microforces into \textquotedblleft slow\textquotedblright\ and
\textquotedblleft fast\textquotedblright\ components. There is no random force
in our approach so we avoid the complications of the conventional Wiener
process in the Langevin approach. Clearly, the deterministic microforces are
oblivious to the stochastic nature of the thermodynamic system. The second law
emerges automatically after averaging, but not without it. Thus, the $\mu$NEQT
as an extension of the MNEQT will be based solely on the sets $\left\{
E_{k}\right\}  $ and $\left\{  p_{k}\right\}  $ so it provides a
\emph{first-principles} deterministic theory from which the MNEQT is trivially reconstructed.

To investigate the ensemble, it is useful to treat a microquantity that takes
values $\left\{  q_{k}\right\}  $ over $\left\{  \mathfrak{m}_{k}\right\}  $
at each instant as a \emph{random variable} $\mathsf{q}$ defined over
$\left\{  \mathfrak{m}_{k}\right\}  $. Thus, $\left\{  E_{k}\right\}  $ and
$\left\{  F_{k,\text{BP}}\right\}  $ refer to the outcomes of the random
variables \textsf{E} and \textsf{F}$_{\text{BP}}$, respectively. In this
study, we use sans serif typeface to denote random variables to distinguish
them from their outcomes. For a given $\left\{  p_{k}\right\}  $, \textsf{q}
is characterized by its ensemble average $\left\langle \mathsf{q}\right\rangle
$ and various moments such as the variance $\left\langle (\Delta
\mathsf{q})^{2}\right\rangle $ in terms of the fluctuation $\Delta
\mathsf{q\doteq q-}\left\langle \mathsf{q}\right\rangle $. As $p_{k}$'s
continue to change in a NEQ\ state, $\left\langle \mathsf{q}\right\rangle $
and $\left\langle (\Delta\mathsf{q})^{2}\right\rangle $\ also change. In the
$\mu$NEQT, the macroforce $\mathbf{F}_{\text{BP}}\doteq\left\langle
\text{\textsf{F}}_{\text{BP}}\right\rangle $ corresponding to $\left\{
\mathbf{F}_{k,\text{BP}}\right\}  $ must oppose the motion in accordance with
the second law as does $\mathbf{F}_{k,\text{f}}(t)$ but not individual
$\mathbf{F}_{k,\text{BP}}$'s. The non-vanishing fluctuations, see Eqs.
(\ref{F-Fluctuation0}) and (\ref{F-Fluctuation}), in $\mathbf{F}_{k,\text{BP}%
}$ even in equilibrium (where $\mathbf{F}_{\text{BP}}=0$) demonstrates that
$\mathbf{F}_{k,\text{BP}}$'s do not always oppose the motion of $\mathfrak{m}%
_{k}$ in the $\mu$NEQT. This effectively means that if we consider
\textsf{F}$_{\text{BP}}$ to be of the form $(-\mathsf{\gamma v})$,
$\mathsf{v,\gamma}$\ having the outcomes $\left\{  v_{k}\right\}  ,\left\{
\gamma_{k}\right\}  $, respectively, then $\gamma_{k}$ is of either sign.

For thermodynamic considerations, instead of considering \textsf{F}%
$_{\text{BP}}$, we will find it convenient to consider the \emph{internal
microwork }$d_{\text{i}}$\textsf{W}$_{\text{BP}}$\ done by it with outcomes
$\left\{  d_{\text{i}}W_{k,\text{BP}}\right\}  $, see Eq.
(\ref{InternalMicrowork-BP}). Being specific to $\mathfrak{m}_{k}$, the
internal microwork $d_{\text{i}}W_{k,\text{BP}}$ also has a unique value but
no specific sign; only the ensemble average $d_{\text{i}}W_{\text{BP}}%
\doteq\left\langle d_{\text{i}}\mathsf{W}_{\text{BP}}\right\rangle \geq0$\ in
accordance with the second law as we will see. That $d_{\text{i}%
}W_{k,\text{BP}}$ and $\gamma_{k}$ have no sign restriction and $\mathbf{F}%
_{k,\text{BP}}$\ does not always oppose motion is the \emph{unique} feature of
our approach.

We consider the two systems (a) and (b) shown in Fig. \ref{Fig_Piston-Spring}
as our system $\Sigma$. In (a), $P,V$ represent some generic work field and
variable, which we label pressure and volume for convenience. We treat the
piston or the particle as a BP. As the BP forms a subsystem, we denote it by
$\Sigma_{\text{BP}}$\ and the remainder of $\Sigma$\ by $\Sigma_{\text{R}}$.
We assume that the piston in (a) may be either mesoscopic or macroscopic,
while the particle in (b) will be assumed to denote a mesoscopic particle.
Thus, our approach will unify the two different scales. We will establish that
both experience fluctuating Brownian motion \textsf{v} over $\left\{
\mathfrak{m}_{k}\right\}  $, except that for the macroscopic size piston, it
is not noticeable because of its macroscopic mass.

We follow Einstein and focus on the BP's center-of-mass. Let $V$ denote the
volume of $\Sigma$ and $\mathbf{P}_{\text{BP}}$ and $\mathbf{P}_{\text{R}}$
the linear momenta of $\Sigma_{\text{BP}}$ and $\Sigma_{\text{R}}$,
respectively. Let $\mathbf{R}_{\text{BP}}$ and $\mathbf{R}_{\text{R}}$ denote
the displacement of the CM of $\Sigma_{\text{BP}}$ and $\Sigma_{\text{R}}$,
respectively. This makes $\Sigma$\ nonuniform and out of EQ \cite{Note}. We
assume $\Sigma$ stationary in the lab-frame (compare with \cite{Gujrati-II})
so that%
\end{subequations}
\begin{equation}
\mathbf{P}_{\text{BP}}+\mathbf{P}_{\text{R}}=0;
\label{Stationary_Momentum_Condition}%
\end{equation}
we also take $\widetilde{\Sigma}$ and, hence, $\Sigma_{0}$ to be stationary so
that $\Sigma$ has no relative motion with respect to $\widetilde{\Sigma}$ and
$\Sigma_{0}$ as noted above.

We will establish here that $\mathbf{P}_{\text{BP}}$ and $\mathbf{P}%
_{\text{R}}$ must be treated as parameters, which is in the spirit of the
original assumption of Einstein about the CM-motion. As $\mathbf{P}%
_{\text{BP}}$ and $\mathbf{P}_{\text{R}}$ denote the total momenta that we
will associate with respective CMs of $\Sigma_{\text{BP}}$ and $\Sigma
_{\text{R}}$, they can only be changed by "external" forces to the two bodies,
\textit{i.e.}, only the force exerted by $\Sigma_{\text{R}}$ on $\Sigma
_{\text{BP}}$ can change $\mathbf{P}_{\text{BP}}$, and the force exerted by
$\Sigma_{\text{BP}}$ on $\Sigma_{\text{R}}$ can change $\mathbf{P}_{\text{R}}%
$. These forces are equal and opposite as they are internal forces for
$\Sigma$, and cancel out in $\Sigma$; recall that it is stationary. Thus, we
need to determine one of these forces in the following. This cancellation also
applies to each microstate of $\Sigma$. However, these "external" forces are
due to some mutual interactions between the two bodies as we discuss at length
in Sec. \ref{Sec-Example}. In the absence of this interaction, $\mathbf{P}%
_{\text{BP}}$ and $\mathbf{P}_{\text{R}}$ cannot change so it is required for
the viscous drag and it cannot be neglected as we have observed above.

We can treat $\Sigma_{\text{R}}$ as our medium $\widetilde{\Sigma}$ and treat
$\Sigma_{\text{BP}}$ as our system $\Sigma$ with $\mathbf{P}_{\text{R}}$
replaced by the linear momentum $\widetilde{\mathbf{P}}$ of $\widetilde
{\Sigma}$ if we want the BP to interact directly with $\widetilde{\Sigma}$, a
case that is a trivial modification but which is usually studied \cite[for
example]{Kapral0}. We will discuss this situation in Sec. \ref{Sec-Example}.%

\begin{figure}
[ptb]
\begin{center}
\includegraphics[
height=1.8273in,
width=3.5492in
]%
{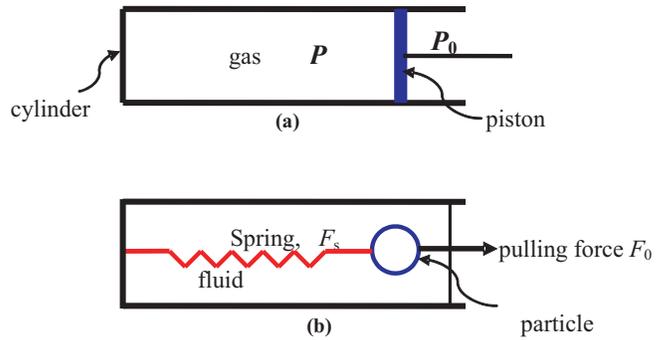}%
\caption{We schematically show a system of (a) gas in a cylinder with a
movable piston (at a distance $l$ from the left wall) under an external
pressure $P_{0\text{ }}$controlling the volume $V$ of the gas and the piston,
and (b) a particle attached to an end of a spring in a fluid and being pulled
by an external force $F_{0}$, which causes the spring to stretch or compress
depending on its direction. The other end of the spring is fixed to the left
wall and $l$ denotes the spring length. The volume of $\Sigma$ in (b) is kept
fixed. In an irreversible process, the internal pressure $P$ or the spring
force $F_{\text{s}}$ is different in magnitude from the external pressure
$P_{0}$ \ or the external force $F_{0}$, respectively. Their difference is the
force imbalance that causes the irreversible macrowork. The temperature of the
system is $T$; $T_{0},P_{0}$ or $F_{0}$ are the macrofields of the medium
$\widetilde{\Sigma}$.}%
\label{Fig_Piston-Spring}%
\end{center}
\end{figure}

\subsection{Deriving Microstate Thermodynamics \label{Sec-MicroTH}}

\subsubsection{Thermodynamic Parameters}

This section is important to demonstrate the importance of relative internal
motion between two parts of a system for viscous dissipation. We will first
treat the piston problem as it is commonly discussed in introductory physics.
The Hamiltonian of the system is written as $\mathcal{H}(\left.
\mathbf{z}\right\vert V,\mathbf{P}_{\text{BP}},\mathbf{P}_{\text{R}})$ in
which $V,\mathbf{P}_{\text{BP}}$ and $\mathbf{P}_{\text{R}}$ form $\mathbf{Z}%
$; here $\mathbf{P}_{\text{BP}}$ and $\mathbf{P}_{\text{R}}$ are two internal
variables. In the following, we treat $\mathbf{z}$\ as discrete and use $k$ as
a label. Let us consider the change
\begin{subequations}
\begin{equation}
d\mathcal{H}=\frac{\partial\mathcal{H}}{\partial\mathbf{z}}\cdot
d\mathbf{z}+\frac{\partial E_{k}}{\partial V}dV+\frac{\partial\mathcal{H}%
}{\partial\mathbf{P}_{\text{BP}}}\cdot d\mathbf{P}_{\text{BP}}+\frac
{\partial\mathcal{H}}{\partial\mathbf{P}_{\text{R}}}\cdot d\mathbf{P}%
_{\text{R}}. \label{HamiltonianChange}%
\end{equation}
The first term on the right vanishes identically due to Hamilton's equations
of motion, so it is the variations due to $d\mathbf{Z}$ ($dV,d\mathbf{P}%
_{\text{BP}}$ and $d\mathbf{P}_{\text{R}}$) that generate any change in
$\mathcal{H}$:%
\begin{equation}
d\mathcal{H}=\frac{\partial E_{k}}{\partial V}dV+\frac{\partial\mathcal{H}%
}{\partial\mathbf{P}_{\text{BP}}}\cdot d\mathbf{P}_{\text{BP}}+\frac
{\partial\mathcal{H}}{\partial\mathbf{P}_{\text{R}}}\cdot d\mathbf{P}%
_{\text{R}}. \label{HamiltonianChange-Work}%
\end{equation}
We identify this as the \emph{generalized} work $dW_{k}$ done by the system
\cite{Gujrati-GeneralizedWork,Gujrati-GeneralizedWork-Expanded,Gujrati-Entropy1,Gujrati-Entropy2}%
. We introduce "generalized mechanical forces" in terms of $E_{k}$ (we
suppress $\mathbf{Z}$ and use $E_{k}$ for $\mathcal{H}(\left.  \mathbf{z}%
\right\vert \mathbf{Z})$ unless clarity is needed) using the standard
definition%
\end{subequations}
\begin{equation}
P_{k}\doteq-\frac{\partial E_{k}}{\partial V},-\mathbf{V}_{k\mathbf{,}%
\text{BP}}\doteq-\frac{\partial E_{k}}{\partial\mathbf{P}_{\text{BP}}%
},-\mathbf{V}_{k\mathbf{,}\text{R}}\doteq-\frac{\partial E_{k}}{\partial
\mathbf{P}_{\text{R}}}; \label{Microforces}%
\end{equation}
these are the conjugate microfields of $V,\mathbf{P}_{\text{BP}}$ and
$\mathbf{P}_{\text{R}}$, respectively. As $E_{k}$ is uniquely determined by
its arguments, these microforces are \emph{deterministic} functions of
$V,\mathbf{P}_{\text{BP}}$ and $\mathbf{P}_{\text{R}}$ and are continuous in a
proper thermodynamic theory; see below. The corresponding generalized
microworks are $P_{k}dV$, etc. so the net microwork done by $\Sigma$ is
\begin{equation}
dW_{k}=P_{k}dV-\mathbf{V}_{k\mathbf{,}\text{BP}}\mathbf{\cdot}d\mathbf{P}%
_{\text{BP}}-\mathbf{V}_{k\mathbf{,}\text{R}}\mathbf{\cdot}d\mathbf{P}%
_{\text{R}}=-dE_{k}. \label{Microwork}%
\end{equation}
The ensemble averages of the various microworks are given by $PdV$, etc., see
Landau and Lifshitz \cite{Landau,Landau-Fluid} and elsewhere \cite{Gujrati-II}%
, where%
\begin{equation}
P\doteq-\partial E/\partial V,\mathbf{V}_{\text{BP}}\doteq\partial
E/\partial\mathbf{P}_{\text{BP}},\mathbf{V}_{\text{R}}\doteq\partial
E/\partial\mathbf{P}_{\text{R}} \label{Macroforces}%
\end{equation}
denote macroforces in the MNEQT; here $E$ is the macroenergy $E=\left\langle
\mathcal{H}(\left.  \mathbf{z}\right\vert V,\mathbf{P}_{\text{BP}}%
,\mathbf{P}_{\text{R}})\right\rangle $ in the lab frame; the conjugate
macrofields are the average pressure $P$ and the average velocities (or
affinities) $\mathbf{V}_{\text{BP}},\mathbf{V}_{\text{R}}$ of the BP and
$\Sigma_{\text{R}}$, respectively, with $E,V,\mathbf{P}_{\text{BP}}%
,\mathbf{P}_{\text{R}}$ forming $\mathfrak{S}$, where the entropy $S$ is
defined as a state function.

We assume that $\Sigma$ is in IEQ \cite{Gujrati-I,Gujrati-II}. Thus, $S$ is a
state function $S(E,V,\mathbf{P}_{\text{BP}},\mathbf{P}_{\text{R}})$ defined
in $\mathfrak{S}$ because of which IEQ macrostates have close similarities
with EQ macrostates so that the temperature $T$ of the system is given by Eq.
(\ref{IEQ-Temp}) and the generalized macroheat by $dQ=TdS$; see below. In
addition, IEQ states have no memory of where they come from. Despite this, IEQ
states have irreversible entropy generation. In EQ, $T=T_{0},P=P_{0}$; see
Fig. \ref{Fig_System}. In addition, $\mathbf{V}_{\text{BP}}$ and
$\mathbf{V}_{\text{R}}$ \emph{vanish} in EQ so they also represent the
vanishing affinities of the medium. As they vanish, they contribute nothing to
the \emph{exchange} microwork $d_{\text{e}}W_{k}$ \cite{deGroot,Prigogine},
which then becomes $d_{\text{e}}W_{k}=P_{0}dV$. \ 

From $E(S,V,\mathbf{P}_{\text{BP}},\mathbf{P}_{\text{R}})$ we have
$dE=TdS-PdV+\mathbf{V}_{\text{BP}}\mathbf{\cdot}d\mathbf{P}_{\text{BP}%
}+\mathbf{V}_{\text{R}}\mathbf{\cdot}d\mathbf{P}_{\text{R}}$, which we rewrite
using Eq. (\ref{Stationary_Momentum_Condition}) as%
\begin{equation}
dE=TdS-PdV+\mathbf{V\cdot}d\mathbf{P}_{\text{BP}}, \label{dE-relation}%
\end{equation}
in terms of the $\emph{relative}$ \emph{velocity} or the \emph{drift velocity
}%
\begin{equation}
\mathbf{V\doteq V}_{\text{BP}}-\mathbf{V}_{\text{R}}=\mathbf{P}_{\text{BP}}/m
\label{Rel-Velocity}%
\end{equation}
of the BP with respect to $\Sigma_{\text{R}}$ in the MNEQT; here $m$ is the
reduced mass of $\Sigma_{\text{BP}}$ and $\Sigma_{\text{R}}$.

We remark that even though $\mathbf{P}_{\text{BP}}$ as the total momentum of
the BP is its intrinsic property, it is coupled to $\Sigma_{\text{R}}$ in
accordance with Eq. (\ref{Stationary_Momentum_Condition}). Consequently,
\begin{equation}
E_{\text{CM}}=\mathbf{P}_{\text{BP}}^{2}/2m \label{CM-KE}%
\end{equation}
is the sum of the kinetic energies of the CM's of $\Sigma_{\text{BP}}$ and
$\Sigma_{\text{R}}$. This is not surprising as the CM-kinetic energies can be
always separated out from the motion of the particles in $\Sigma$. As
discussed in Sec. \ref{Sec-IEQ}, it is possible to have a different
temperature $T_{\text{CM}}$ associated with the CM-motion, which can be very
different from $T$. As this motion slows down, $T_{\text{CM}}$ will continue
to decrease; cf. Eq. (\ref{Irreversible-BP-Energy-Temperature}).

\subsubsection{Einstein-Langevin Duality of the Relative Motion}

We can also rewrite the drift velocity term using the identity
\begin{equation}
\mathbf{V\cdot}d\mathbf{P}_{\text{BP}}\equiv\mathbf{F}_{\text{BP}%
}\mathbf{\cdot}d\mathbf{R=}d(\mathbf{P}_{\text{BP}}^{2}/2m),
\label{V-F_BP-Relation}%
\end{equation}
where $\mathbf{F}_{\text{BP}}\doteq d\mathbf{P}_{\text{BP}}\mathbf{/}dt$ is
the "external" \emph{macroforce}\ as discussed above, and $d\mathbf{R=V}dt$ is
the \emph{relative displacement} of the BP in the MNEQT. Because of this
identity, we can either use $\mathbf{P}_{\text{BP}}$ or $\mathbf{R}$ as a
parameter in $\mathbf{Z}$ so\ the macroenergy $E$ can be expressed either as
$E_{\mathbf{P}}\doteq E(S,V,\mathbf{P}_{\text{BP}})$ or $E_{\mathbf{R}}\doteq
E(S,V,\mathbf{R})$, a simplification due to the thermodynamic treatment, with
\begin{equation}
\mathbf{V}=\partial E(S,V,\mathbf{P}_{\text{BP}})/\partial\mathbf{P}%
_{\text{BP}},\mathbf{F}_{\text{BP}}=\partial E(S,V,\mathbf{R})/\partial
\mathbf{R.} \label{V-F_Definition}%
\end{equation}
We now deal with a reduced state space $\mathfrak{S}^{\prime}$ formed by
$E,V,\mathbf{P}_{\text{BP}}$ or $E,V,\mathbf{V}$. In a proper thermodynamic
theory, $E$ is at least twice differentiable (we do not consider any phase
transition in this work) so the above derivatives exist and are continuous.

There is very interesting duality hidden in Eq. (\ref{V-F_BP-Relation}). The
choice of using $\mathbf{R}$ as a parameter provides a justification for
Einstein's approach involving the CM location of the BP and considering the
"osmotic" force $\mathbf{F}_{\text{BP}}$ acting on it; there was no need to
consider its momentum at all. Thus, his choice in our approach corresponds to
using $E$ as $E(S,V,\mathbf{R})$. On the other hand, Langevin's interest was
not in using $\mathbf{R}$ but its momentum $\mathbf{P}_{\text{BP}}$ to write
down the equation of motion; cf. Eq. (\ref{LangevinEquation}). While he was
not interested in thermodynamics, his choice in our approach will correspond
to using $E$ as $E(S,V,\mathbf{P}_{\text{BP}})$. As a consequence, our
thermodynamic approach is a hybrid approach capable of allowing both
approaches in a unifying way. However, as the first equation in Eq.
(\ref{V-F_Definition}) merely gives back $\mathbf{P}_{\text{BP}}=m\mathbf{V}$,
it is not much of a use. Therefore, we will normally use $E(S,V,\mathbf{R})$
with $\mathbf{R}$ as a parameter, which will be extremely useful in our
thermodynamic investigation.

\subsubsection{Microwork and Microheat}

Using the generalized macrowork $dW=PdV-\mathbf{F}_{\text{BP}}\mathbf{\cdot
}d\mathbf{R}$ and macroheat $dQ=TdS$, we have $dE=dQ-dW$, which expresses the
first law in terms of the generalized quantities. This expresses an important
fact: the two terms in it denote\emph{ independent} variations of the energy
$E$: $dQ$ denotes the change due to entropy variation and $dW$ isentropic
variation. This allows us to deal with $dW$ as a purely mechanical ($dS=0$)
quantity resulting in microstate energy changes. This is easily seen from the
following argument. From $E\equiv\left\langle \mathsf{E}\right\rangle \doteq%
{\textstyle\sum\nolimits_{k}}
E_{k}p_{k}$ in terms of $E_{k}=E_{k}(V,\mathbf{P}_{\text{BP}})$ [or
equivalently $E_{k}=E_{k}(V,\mathbf{R})$] and $p_{k}$, we have
\[
dE=%
{\textstyle\sum\nolimits_{k}}
E_{k}dp_{k}+%
{\textstyle\sum\nolimits_{k}}
p_{k}dE_{k},
\]
where
\[
dE_{k}=(\partial E_{k}/\partial V)dV+(\partial E_{k}/\partial\mathbf{P}%
_{\text{BP}})\mathbf{\cdot}d\mathbf{P}_{\text{BP}}.
\]
The first sum in $dE$ involves $dp_{k}$ at fixed $E_{k}$, and evidently
corresponds to the entropy change $dS$. It denotes the generalized heat
\[
dQ=\left\langle d\mathsf{Q}\right\rangle =%
{\textstyle\sum\nolimits_{k}}
p_{k}dQ_{k}\doteq%
{\textstyle\sum\nolimits_{k}}
E_{k}dp_{k}.
\]
Here, we have used $d\mathsf{Q}$ to denote a random variable with outcomes
$\left\{  dQ_{k}\right\}  $. The second sum in $dE$ involves $dE_{k}$ at fixed
$p_{k}$ and evidently corresponds to $dS=0$. Its negative is the generalized
work (we use $d\mathsf{W}$ to denote a random variable with outcomes $\left\{
dW_{k}\right\}  $)
\[
dW=\left\langle d\mathsf{W}\right\rangle \doteq-%
{\textstyle\sum\nolimits_{k}}
p_{k}dE_{k},
\]
and \emph{uniquely} identifies microwork $dW_{k}=-dE_{k}$ from which we can
uniquely identify mechanical microforces $P_{k}=-(\partial E_{k}/\partial V)$
and $(-\mathbf{V}_{k})=-\partial E_{k}/\partial\mathbf{P}_{\text{BP}}$ that
appear in the $\mu$NEQT; these quantities refer to the system alone. This can
be done because $dW_{k}$ is a mechanical quantity and is oblivious to $p_{k}$.

We use this uniqueness of identifying system-specific microforces to construct
the $\mu$NEQT in $\mathfrak{S}^{\prime}$. We have
\begin{subequations}
\begin{equation}
dW_{k}=P_{k}dV-\mathbf{V}_{k}\mathbf{\cdot}d\mathbf{P}_{\text{BP}}\equiv
P_{k}dV-\mathbf{F}_{k\text{,BP}}\mathbf{\cdot}d\mathbf{R}%
,\label{Microwork-Reduced}%
\end{equation}
where
\begin{equation}
\mathbf{V}_{k}\doteq\partial E_{k}/\partial\mathbf{P}_{\text{BP}}%
,\mathbf{F}_{k\text{,BP}}\doteq\partial E_{k}/\partial\mathbf{R}%
.\label{MicroForce-Velocity}%
\end{equation}
Using $d_{\text{e}}W=P_{0}dV$, we identify the \emph{irreversible} macrowork
$d_{\text{i}}W\doteq dW-d_{\text{e}}W$%
\end{subequations}
\begin{equation}
d_{\text{i}}W=(P-P_{0})dV-\mathbf{F}_{\text{BP}}\mathbf{\cdot}d\mathbf{R\geq
}0\label{Irreversible Work}%
\end{equation}
from the second law so that we must have%
\begin{equation}
(P-P_{0})dV\geq0,\mathbf{F}_{\text{BP}}\mathbf{\cdot}d\mathbf{R}%
\leq0\label{SecondLawConsequences}%
\end{equation}
separately as each term refers to an \emph{independent} internal process. For
the example in Fig. \ref{Fig_Piston-Spring}(b), we must replace $(P-P_{0})dV$
by $(F_{\text{s}}-F_{0})dl$, where $dl$ is the spring compression. Similarly,
the exchange heat with $\widetilde{\Sigma}$\ is $d_{\text{e}}Q=T_{0}%
d_{\text{e}}S$ and the irreversible heat is
\begin{equation}
d_{\text{i}}Q=TdS-T_{0}d_{\text{e}}S=(T-T_{0})d_{\text{e}}S+Td_{\text{i}%
}S~\mathbf{\geq}0.\label{diQ}%
\end{equation}
As $dE=d_{\text{e}}Q-d_{\text{e}}W$ also expresses the first law, we must
have
\begin{equation}
d_{\text{i}}Q=d_{\text{i}}W\ \mathbf{\geq}0\label{diQ-diW}%
\end{equation}
in the MNEQT. Therefore, determining $d_{\text{i}}W$ allows us to indirectly
determine $d_{\text{i}}Q$. In this study, we will not be directly studying
generalized heat, which we will consider in a future publication.

We thus see that the $\mu$NEQT is obtained directly and uniquely from the
MNEQT. However, the most important and distinguishing feature of our approach
as noted above is that the microwork $dW_{k}$ is deterministic (independent of
the probability $p_{k}$) so it represents a truly microscopic mechanical work
from which we can directly identify various microscopic forces.\ Thus, even
though we have started with the MNEQT, the microscopic work $dW_{k}$ in Eq.
(\ref{Microwork-Reduced}) directly and uniquely identifies microscopic forces
$P_{k}$ and $\mathbf{V}_{k}$ in terms of purely mechanical quantities of the
system alone. As we will see, the $\mu$NEQT provides additional details than
are not available from using the MNEQT alone.

\subsubsection{IEQ Microstate Probabilities}

In an IEQ state \cite{Gujrati-Entropy-Note}, we have two possible forms of
$p_{k}$ based on the choice of the parameters $\mathbf{R}$ or $\mathbf{F}%
_{\text{BP}}$:
\begin{align}
p_{k}  &  =\exp[\Phi-(E_{k}+P_{k}V-\mathbf{F}_{k,\text{BP}}\mathbf{\cdot
R})]/T],\label{IEQ-probabilities}\\
p_{k}  &  =\exp[\Phi-(E_{k}+P_{k}V-\mathbf{F}_{\text{BP}}\mathbf{\cdot R}%
_{k})]/T], \label{IEQ-probabilities0}%
\end{align}
with $\left\langle \mathsf{1}\right\rangle ,\left\langle \mathsf{E}%
\right\rangle ,\left\langle \mathsf{P}\right\rangle $, and $\left\langle
\mathsf{F}_{\text{BP}}\right\rangle \ $in Eq.(\ref{IEQ-probabilities}) or
$\left\langle \mathsf{R}\right\rangle $ in Eq.(\ref{IEQ-probabilities0}),
fixed so that $\beta=1/T,\beta V$ and $(-\beta\mathbf{R)}$ or $(-\beta
\mathbf{F}_{\text{BP}}\mathbf{)}$ are Lagrange multipliers to maximize the
entropy \cite{Gujrati-Entropy2}. Here, the normalization function $\Phi
$\ ensures that $p_{k}$'s add to unity. The form is what is expected in EQ
except for the presence of the internal variable term and of the fields $T$
and $P$ of the IEQ state. Thus, most of the EQ results can be easily extended
to an IEQ state.

We now prove a very useful and general theorem for systems in IEQ that allows
us to identify the change in the IEQ temperature as its parameters change.

\begin{theorem}
\label{Th-Temp-Change} As the parameters in $\mathbf{Z}$ change and change the
microstate probabilities, the change in the temperature is given by%
\begin{equation}
dT=T\frac{\left\langle d\Psi\right\rangle }{\left\langle \Psi\right\rangle }
\label{dT0}%
\end{equation}
where we have introduced%
\begin{subequations}
\begin{equation}
\Psi_{k}=\Phi-(E_{k}+P_{k}V-\mathbf{F}_{k,\text{BP}}\mathbf{\cdot
R})\mathbf{,} \label{Psi-function-k}%
\end{equation}

\end{subequations}
\end{theorem}

\begin{proof}
Using
\begin{subequations}
\begin{equation}
p_{k}=\exp(\Psi_{k}/T), \label{p_k-Psi}%
\end{equation}
we find that
\begin{equation}
dp_{k}=p_{k}[\frac{d\Psi_{k}}{T}-\frac{\Psi_{k}dT}{T^{2}}]. \label{dp_k-Psi}%
\end{equation}
The average $\left\langle \Psi\right\rangle $ is given by%
\end{subequations}
\begin{subequations}
\begin{equation}
\left\langle \Psi\right\rangle =\Phi-E-PV+\mathbf{F}_{\text{BP}}\mathbf{\cdot
R,} \label{Psi-Average}%
\end{equation}
and $\left\langle d\Psi\right\rangle $ is given by
\begin{equation}
\left\langle d\Psi\right\rangle =d\Phi+dW-d(PV)+d(\mathbf{F}_{\text{BP}%
}\mathbf{\cdot R).} \label{dPsi-Average}%
\end{equation}
Eq. (\ref{dT0}) now follows from $%
{\textstyle\sum\nolimits_{k}}
dp_{k}=0$ as an identity for any IEQ macrostate.
\end{subequations}
\end{proof}

To use Eq. (\ref{dT0}), we must explicitly evaluate $\left\langle
\Psi\right\rangle $ and $\left\langle d\Psi\right\rangle $\ using Eqs.
(\ref{Psi-Average}) and (\ref{dPsi-Average}), respectively, in terms of
quantities appearing on their right sides.

\subsection{Viscous Drag and the Langevin Limit\label{Sec-ViscousDrag}\ \ \ }

We use the notation
\begin{equation}
d_{\text{i}}W_{\text{BP}}\doteq-\mathbf{F}_{\text{BP}}\mathbf{\cdot
}d\mathbf{R}\equiv-\mathbf{V\cdot}d\mathbf{P}_{\text{BP}}\geq
0\label{Irreversible Friction Work}%
\end{equation}
related to the second irreversible contribution in Eq.
(\ref{SecondLawConsequences}). It follows that for the inequality to be valid,
we must have the following form for NEQ $\mathbf{F}_{\text{BP}}$ in the MNEQT%
\begin{equation}
\mathbf{F}_{\text{BP}}=-\mathbf{V}f(T,\mathbf{V},t),f(T,\mathbf{V}%
,t)>0,\label{FrictionForm}%
\end{equation}
in which $f(T,\mathbf{V},t)$ must be an even scalar function of $\mathbf{V=P}%
_{\text{BP}}/m$\ at each instant so that $d_{\text{i}}W_{\text{BP}%
}=f(T,\mathbf{V},t)\mathbf{V}^{2}dt\geq0$. As $\mathbf{F}_{\text{BP}}$ opposes
motion, it represents the viscous force we are interested in. Let us compare
$\mathbf{F}_{\text{BP}}$ above with its definition in Eq.
(\ref{V-F_Definition}), according to which it is a derivative of $E$ with
respect to $\mathbf{R}$. The only way this derivative can give a result along
the direction of $\mathbf{V}$ is for the scalar function $E$ to be a function
of the combination
\begin{equation}
u\doteq\mathbf{V}\cdot\mathbf{R}\label{ScalarCombination}%
\end{equation}
as a scalar. We thus conclude that
\begin{equation}
f(T,\mathbf{V},t)=-\partial E(S,V,u)/\partial u,\label{General-friction-f}%
\end{equation}
so $f(T,\mathbf{V},t)$ will also include a dependence on $\mathbf{R}$ in $f$
through $u$ so we must write it as $f(T,V,u)$. Hopefully, this will make Eq.
(\ref{FrictionForm}) suitable for some active BPs
\cite{Marconi,Romanczuk,Kapral-2017,Fodor}. In general, the dependence on
$\mathbf{V}$\ through $u$\ may be very complex as will become clear in Sec.
\ref{Sec-Example}.

As $\mathbf{F}_{\text{BP}}$\ is the macroforce corresponding to the viscous
drag, the above discussion provides a \emph{thermodynamic justification} of
the viscous drag. To make connection with the Langevin equation, we will
assume $f(u,t)$ to be a power series in $u^{2}\ $with $f(0,t)=\gamma(t)\geq0$
so that $d_{\text{i}}W_{\text{BP}}\simeq\gamma(t)\mathbf{V\cdot}d\mathbf{R}$
is the \emph{frictional work} in the small-speed approximation, which will be
called the \emph{Langevin limit} from now on. In this limit, $d_{\text{i}%
}W_{\text{BP}}\simeq\gamma(t)\mathbf{V}^{2}(t)dt\geq0$ at any instant $t$.
Langevin takes $\gamma(t)$ to be a constant $\gamma$.

The above discussion also provides a thermodynamic justification of the
viscous drag in the Langevin equation in the small-speed approximation. For
arbitrary speeds, we can treat $f(T,u,t)$ as the analog of an \emph{effective}
$\gamma_{\text{eff}}$ in Eq. (\ref{LangevinEquation}), which is a complicated
function of $T,\mathbf{V}$ and $t$, a situation commonly encountered in active
BPs \cite{Marconi,Romanczuk,Kapral-2017,Fodor}. We will not pursue active BPs
in this work except tangentially; they will be treated later.

As we will see below, we get more insight into the viscous force
$\mathbf{F}_{\text{BP}}(t)$ when we consider its microanalogs $\mathbf{F}%
_{k,\text{BP}}(\mathbf{V,}t)$ in the $\mu$NEQT.

\section{Fluctuations and A New Equation of Motion
\label{Sec-InternalContribution}}

The fluctuations in random variables are the hallmark of a statistical system
and are always present whether we consider a reversible or an irreversible
process. Let us consider the random variable \textsf{P}\ with outcomes
$\left\{  P_{k}\right\}  $. The fluctuation $\Delta$\textsf{P }has outcomes
$\left\{  P_{k}-P\right\}  $ with $P=\left\langle \mathsf{P}\right\rangle $,
which determine the mean square fluctuation $\left\langle (\Delta
\mathsf{P})^{2}\right\rangle \geq0$. We know from EQ statistical mechanics
($P=P_{0},T=T_{0}$) \cite{Landau} that
\[
\left\langle (\Delta\mathsf{P})^{2}\right\rangle _{\text{eq}}=-T_{0}(\partial
P/\partial V)_{S}%
\]
is not identically zero so $P_{k}$ fluctuates over $\mathfrak{m}_{k}$ and
takes values on both sides of $P_{0}$. Since \textsf{P} is not determined by
any macrostate, $\left\{  P_{k}\right\}  $ remain the same whether we are
dealing with an EQ or a NEQ macrostate. Moreover, $d_{\text{i}}W_{k,V}%
\doteq(P_{k}-P_{0})dV$ does not have a particular sign in general, even though
the macrowork $d_{\text{i}}W_{V}\doteq(P-P_{0})dV\geq0$ is never negative; see
Eq. (\ref{SecondLawConsequences}). Because of this conformity, it is customary
to call the macrowork $d_{\text{i}}W_{V}$ the \emph{irreversible} work. As the
microwork $d_{\text{i}}W_{k,V}$ does not follow the sign requirement, it is
better to call it internal microwork as noted above.

Similarly, there are fluctuations in the \emph{random variables}
$\mathsf{F}_{\text{BP}}$ and $\mathsf{V}$ (with outcomes $\left\{
\mathbf{F}_{k\text{,BP}}\right\}  $ and $\left\{  \mathbf{V}_{k}\right\}  $,
respectively) around the average $\mathbf{F}_{\text{BP}}\equiv\left\langle
\mathsf{F}_{\text{BP}}\right\rangle $ and $\mathbf{V}\equiv\left\langle
\mathsf{V}\right\rangle $, respectively, which are always present. This will
be explicitly demonstrated later; see Eqs. (\ref{F-Fluctuation0}%
-\ref{F-Fluctuation}). As the EQ affinity $\mathbf{F}_{0\text{BP}}=0$ or
$\mathbf{V}_{0}=0$\ so that $d_{\text{e}}W_{k\text{,BP}}\equiv0$
\cite{deGroot,Prigogine,Maugin}, we conclude that $d_{\text{i}}W_{k\text{,BP}%
}\equiv dW_{k\text{,BP}}$ fluctuates over $\left\{  \mathfrak{m}_{k}\right\}
$ around the macroaverage $dW_{\text{BP}}\equiv d_{\text{i}}W_{\text{BP}}$.
Thus, the internal microwork $d_{\text{i}}W_{k\text{,BP}}$ does not have a
particular sign, while $d_{\text{i}}W_{\text{BP}}$ does as seen in Eq.
(\ref{Irreversible Friction Work}).

It is important to make the following three remarks concerning $\mathbf{F}%
_{k,\text{BP}}(t)$:

\begin{enumerate}
\item[(a)] It is not broken into a fast- and a slow-component for each $k$ as
is common in the Langevin approach.

\item[(b)] It represents the outcome of a random variable $\mathsf{F}%
_{\text{BP}}$ over the microstates.

\item[(c)] For a given $k$, $\mathsf{F}_{\text{BP}}$ possesses no randomness
so $\mathbf{F}_{k,\text{BP}}(t)$ has a unique value.
\end{enumerate}

\subsection{A New Equation of Motion}

That the internal microwork $d_{\text{i}}W_{k\text{,BP}}$ has no sign
restriction is another point of departure from Langevin's approach and is
discussed next. We focus on the form%
\begin{equation}
d_{\text{i}}W_{k\text{,BP}}\doteq-\mathbf{F}_{k\text{,BP}}\mathbf{\cdot
}d\mathbf{R}_{k} \label{InternalMicrowork-BP}%
\end{equation}
for $\mathfrak{m}_{k}$ and determine Newton's equation for the BP at a
relative location $\mathbf{R}_{k}(t)$; $k$ on $\mathbf{R}_{k}(t)$\ is added
for clarity. The \emph{deterministic} equation
\begin{equation}
md^{2}\mathbf{R}_{k}(t)/dt^{2}=md\mathbf{V}_{k}(t)/dt=\mathbf{F}_{k,\text{BP}%
}(t) \label{NewLangevinEq}%
\end{equation}
describes the trajectory of the BP in the $\mu$NEQT. The trajectory
$\mathbf{R}_{k}(t)$ is obtained by integrating twice Eq. (\ref{NewLangevinEq})
using basic calculus, and is also deterministic and at least twice
differentiable. Introducing the deviation $\Delta\mathbf{F}_{k,\text{BP}%
}(t)\doteq\mathbf{F}_{k,\text{BP}}(t)-\mathbf{F}_{\text{BP}}(t)$, we can
express $\mathbf{F}_{k,\text{BP}}(t)$ in terms of $\mathbf{F}_{\text{BP}}(t)$
as%
\begin{equation}
\mathbf{F}_{k,\text{BP}}(t)\doteq\mathbf{F}_{\text{BP}}(t)+\Delta
\mathbf{F}_{k,\text{BP}}(t), \label{NewLangevinForce}%
\end{equation}
which may suggest that $\Delta\mathbf{F}_{k,\text{BP}}(t)$ is Langevin's
$\boldsymbol{\xi}(t)$. This is where other important differences from the
Langevin approach appear. The $\mathbf{F}_{\text{BP}}(t)$ is a function of the
average relative velocity $\mathbf{V}(t)\equiv\left\langle \mathsf{V}%
(t)\right\rangle $ so it does not represent the microforce $\mathbf{F}%
_{k,\text{f}}(t)\ $that appears in Eqs. (\ref{LangevinEquation0}%
-\ref{LangevinEquation}). Furthermore, $\mathbf{F}_{k\text{,BP}}\doteq\partial
E_{k}/\partial\mathbf{R}$ is deterministic in the $\mu$NEQT as noted above. So
is $\mathbf{F}_{\text{BP}}(t)$. Thus, $\Delta\mathbf{F}_{k,\text{BP}}(t)$ also
takes a \emph{single} value for each $\mathfrak{m}_{k}$,\ while
$\boldsymbol{\xi}(t)$ is stochastic. The stochasticity in the $\mu$NEQT
emerges as we average Eq. (\ref{NewLangevinEq}) over all microstates to yield
\begin{equation}
md^{2}\mathbf{R}(t)/dt^{2}=\mathbf{F}_{\text{BP}}(t)-2m\left\langle \dot
{p}\mathbf{\dot{R}}/p\right\rangle -m\left\langle \ddot{p}\mathbf{R}%
/p\right\rangle , \label{NewLangevinEqAv}%
\end{equation}
with $\mathbf{R}(t)\doteq\left\langle \mathbf{R}(t)\right\rangle ,p_{k}>0$,
and a dot represents the total time derivative; the last two terms on the
right side are due to temporal changes in $\left\{  p_{k}\right\}  $; they
vanish in EQ so that we obtain a simple equation of motion for the average
trajectory $\mathbf{R}(t)$ of $\mathbf{R}_{k}(t)$\ that is normally discussed
in the literature for the Langevin equation.

\subsubsection{Solving Eq. (\ref{NewLangevinEq}) \label{Sec-Soln-Eq-Motion}}

We will consider the simpler case by keeping $V$ constant so we do not have to
worry about the $PV$-work. Let $\mathbf{R}_{k}(0)$ and $\mathbf{V}_{k}(0)$ be
the initial values of $\mathbf{R}_{k}(t)$ and $\mathbf{V}_{k}(t)$,
respectively; let $T(0)$ be the initial value of $T$. It is convenient to
discretize the situation by dividing a predetermined time interval $\Delta t$,
over which we are interested in finding the solution, into $n$ nonoverlapping
intervals $\delta t_{l-1}=t_{l}-t_{l-1},l=1,\cdots,n,$ with $t_{0}=0$ and
$t_{n}=\Delta t$. We determine the initial value $\mathbf{F}_{k,\text{BP}}(0)$
of $\mathbf{F}_{k,\text{BP}}(t)$ using Eq. (\ref{MicroForce-Velocity}) and use
Eq. (\ref{IEQ-probabilities}) to determine the initial probability $p_{k}(0)$;
note that we must not consider the $P_{k}V$-term for this case. We now solve
Eq. (\ref{NewLangevinEq}) during $\delta t_{0}$\ to determine the next values
of $\mathbf{R}_{k}(t_{1})$, which is then used to determine the next values of
$\mathbf{F}_{k,\text{BP}}(t_{1})$ and $p_{k}(t_{1})$. We repeat these steps
$n$ times to obtain the solution over the interval $(t_{n},t_{0})=\Delta t$.

\subsubsection{Entropy and Temperature Changes}

The instantaneous macroentropy is given by $S\doteq%
{\textstyle\sum\nolimits_{k}}
p_{k}s_{k},s_{k}\doteq-\ln p_{k}$. Using $dp_{k}(t_{l-1})=p_{k}(t_{l}%
)-p_{k}(t_{l-1})$ obtained above, we determine $dS(t_{l-1})\doteq%
{\textstyle\sum\nolimits_{k}}
dp_{k}(t_{l-1})(s_{k}(t_{l-1})-1)$ and $dQ(t_{l-1})=T(t_{l-1})dS(t_{l-1})$\ in
the MNEQT. Using $d_{\text{e}}Q(t_{l-1})=\widetilde{C}(T_{0}-T(t_{l-1}))$,
where $\widetilde{C}$ is the heat capacity of $\widetilde{\Sigma}$, and
equating it with $T_{0}d_{\text{e}}S(t_{l-1})$, we determine $d_{\text{e}%
}S(t_{l-1})$, which is then used\ to determine the irreversible macroentropy
generation $d_{\text{i}}S(t_{l-1})=dS(t_{l-1})-d_{\text{e}}S(t_{l-1})$.

The temperature change during $\delta t_{l-1}$ is given in Theorem
\ref{Th-Temp-Change}. We thus have a complete MNEQT.

\subsection{Fluctuations\label{Sec-Fluctuations}}

Standard fluctuation theory \cite{Landau,Gujrati-Fluctuations,Mishin} deals
with EQ fluctuations where no internal variables are present. However, as we
have shown elsewhere \cite{Gujrati-Entropy2} and also discussed above, their
presence in IEQ states causes no new complications and we can just follow the
standard formulation to obtain \emph{instantaneous} fluctuations in
$\mathbf{F}_{\text{BP}},\mathbf{R,V}$ and $\mathbf{P}_{\text{BP}}$ when the
system is in an IEQ state involving the internal variable $\mathbf{F}%
_{\text{BP}}$ or $\mathbf{P}_{\text{BP}}$.

We restrict ourselves to a $1$-d case for simplicity ($R$ replaced by $X$).
The probability of fluctuations about the IEQ state \cite{Landau} is given by
$W_{0}\exp(-\beta\rho/2)$, where
\[
\rho=\Delta T\Delta S-\Delta P\Delta V+\Delta F_{\text{BP}}\Delta X
\]
in terms of various fluctuations from the IEQ state and $W_{0}$ is some
unimportant constant. As $\rho$ is a thermodynamic expression, we have the
liberty to chose $T,V$ and $F_{\text{BP}}$ as independent variables to express
$\rho$ in terms of $\Delta T,\Delta V$ and $\Delta F_{\text{BP}}$%
:$\rho=(\partial S/\partial T)(\Delta T)^{2}-(\partial P/\partial V)(\Delta
V)^{2}+2(\partial X/\partial T)\Delta T\Delta F_{\text{BP}}+(\partial
X/\partial F_{\text{BP}})(\Delta F_{\text{BP}})^{2}$ by exploiting some
Maxwell relations \cite{Gujrati-III}. The coefficients of fluctuations in the
$\Delta T$-$\Delta F_{\text{BP}}$ subspace define a $2\times2$ matrix
$\mathbf{M}$ from which we can determine various mean square fluctuations
\cite{Landau,Gujrati-Fluctuations,Mishin}. For the interesting mean square
fluctuation $\langle(\Delta\mathsf{F}_{\text{BP}})^{2}\rangle$, we obtain
\begin{equation}
\langle(\Delta\mathsf{F}_{\text{BP}})^{2}\rangle=T(\partial S/\partial
T)_{V,F_{\text{BP}}}/\mathit{M}, \label{F-Fluctuation0}%
\end{equation}
where $\mathit{M}\doteq$ $(\partial S/\partial T)(\partial X/\partial
F_{\text{BP}})-(\partial X/\partial T)^{2}\geq0$ is the determinant of
$\mathbf{M}$. From these fluctuations, we can determine any other fluctuation
such as $\langle(\Delta X)^{2}\rangle$. However, a simple method is to use
$T,V$ and $X$ as independent variables, which yields%
\[
\langle(\Delta\mathsf{X})^{2}\rangle=-T(\partial P/\partial V)_{T,X}%
/\mathit{M}^{\prime},
\]
where $\mathit{M}^{\prime}\doteq$ $-(\partial P/\partial V)(\partial
F_{\text{BP}}/\partial X)-(\partial F_{\text{BP}}/\partial V)^{2}\geq0$.

\subsection{Microwork Fluctuations}

The above NEQ fluctuation calculation is valid in general for small
fluctuations about some IEQ state and are by very nature Gaussian. To go
beyond the Gaussian form, we must expand to higher order, which we will not do
here as we are only interested in establishing the feasibility of the $\mu
$NEQT and the reproducibility of known results. As $\mathbf{F}_{k\text{,BP}}$
is specific to $\mathfrak{m}_{k}$, $d_{\text{i}}W_{k\text{,BP}}$ is not
affected by $p_{k}$; it is the same whether we consider an EQ or a NEQ state.
In the present case, $d_{\text{i}}W_{k\text{,BP}}\doteq-F_{k\text{,BP}}dX$
fluctuates around its average $d_{\text{i}}W_{\text{BP}}\doteq\left\langle
d_{\text{i}}\mathsf{W}_{\text{BP}}\right\rangle \geq0$. As the average
fluctuation $\langle(\Delta F_{\text{BP}})^{2}\rangle(dX)^{2}$ does not
necessarily vanish, $d_{\text{i}}W_{k\text{,BP}}$ takes values on both sides
of $d_{\text{i}}W_{\text{BP}}$. To understand its variation, we consider the
EQ state for which $d_{\text{i}}W_{\text{BP,eq}}\equiv0$ so that $d_{\text{i}%
}W_{k\text{,BP}}$ takes both positive and negative values around $0$. This
variation remains true even in a NEQ state; only $p_{k}$'s change.

Thus, we have finally established that there is no sign restriction. Having no
restriction on the sign of $d_{\text{i}}W_{k\text{,BP}}$ means that
$\mathbf{F}_{k\text{,BP}}$ may or may not oppose the motion for $\mathfrak{m}%
_{k}$. This is different from the Langevin approach. In the latter, the
deterministic force $\mathbf{F}_{k,\text{f}}(t)\doteq-\gamma\mathbf{v}%
_{k}(t)\neq0$ \emph{always} opposes the motion for every microstate
$\mathfrak{m}_{k}$; this is in accordance with the second law. Therefore,
$d_{\text{i}}W_{k\text{,f}}\doteq\mathbf{F}_{k,\text{f}}(t)\cdot d\mathbf{R}$
must generate some irreversible entropy $d_{\text{i}}S>0$; cf. Eqs.
(\ref{Irreversible Work}) and (\ref{diQ}). In our theory, $\mathbf{F}%
_{k\text{,BP}}$ is a mechanical force so it does not change $p_{k}$ and,
hence, the entropy. The other difference is the following. As $\mathbf{F}%
_{k,\text{f}}(t)$ is obtained from $\mathbf{F}_{k\text{,BP}}^{\prime}$ by
averaging over the fast Langevin force $\boldsymbol{\xi}$, it is analogous is
some crude sense to our deterministic microforce $\mathbf{F}_{k\text{,BP}}$
but the latter does not always oppose the motion. Recall that $\mathbf{F}%
_{k,\text{f}}(t)$ represents the slow component of the microforce on BP, while
$\mathbf{F}_{k\text{,BP}}$\ is the net microforce on BP. If we insist on using
the Langevin interpretation for $\gamma$, then this is equivalent to allowing
$\gamma$ to have both signs as is considered to be the case for active BPs
\cite{Marconi,Romanczuk,Kapral-2017,Fodor}.

\subsection{The Langevin Limit}

If $(\partial S/\partial F_{\text{BP}})_{T,V}=(\partial X/\partial
T)_{V,F_{\text{BP}}}$ can be neglected, the fluctuations in $T,V$ and
$F_{\text{BP}}$ become independent. In particular, $\langle(\Delta
\mathsf{F}_{\text{BP}})^{2}\rangle=T(\partial F_{\text{BP}}/\partial X)_{T,V}%
$. It follows from Eq. (\ref{V-F_Definition}) that $F_{\text{BP}}(S,V,X)$ is a
function of $X$, and using $F_{\text{BP}}\simeq-\gamma\dot{X}$ in the Langevin
limit, we find that $\partial F_{\text{BP}}/\partial X\simeq-\gamma\ddot
{X}/\dot{X}=\gamma^{2}/m$ so that
\begin{equation}
\left\langle (\Delta\mathsf{F}_{\text{BP}})^{2}\right\rangle \simeq
T_{0}\gamma^{2}/m>0, \label{F-Fluctuation}%
\end{equation}
which is precisely what we expect in this approximation since $(\Delta
F_{\text{BP}})^{2}=\gamma^{2}\dot{X}^{2}$ and $\langle\dot{X}^{2}\rangle$
$=T_{0}/m$ in EQ; see below. We can similarly obtain $\left\langle (\Delta
X)^{2}\right\rangle =T_{0}(\partial X/\partial F_{\text{BP}})_{T,V}%
=mT_{0}/\gamma^{2}>0$ and $\left\langle \Delta X\Delta\mathsf{F}_{\text{BP}%
}\right\rangle =T_{0}$. In a highly viscous environment, the mean square
CM-fluctuation becomes very small as expected, and $\langle(\Delta
F_{\text{BP}})^{2}\rangle$ become large. All these results are valid for any
BP of any reduced mass $m$ ranging from mesoscales to macroscales in this limit.

\subsection{Relative Velocity Fluctuations}

Integrating Eq. (\ref{InternalMicrowork-BP}) over an interval $(0,t)$, we have%
\begin{subequations}
\begin{equation}
\Delta_{\text{i}}W_{k\text{,BP}}=-%
{\textstyle\int\nolimits_{0}^{t}}
\mathbf{F}_{k,\text{BP}}(t)\mathbf{\cdot V}_{k}(t)dt,
\label{MicroWork-KineticEnergy1}%
\end{equation}
where $\mathbf{V}_{k}(t)$\ is the relative velocity. Note that we do not need
$p_{k}$\ to calculate the microwork, which makes it trivial; see Sec.
\ref{Sec-BP-N=1}. We thus have the general result
\begin{equation}
\Delta_{\text{i}}W_{k\text{,BP}}=-(m/2)(\mathbf{V}_{k}^{2}(t)-\mathbf{V}%
_{k}^{2}(0)), \label{MicroWork-KineticEnergy2}%
\end{equation}
which can have any sign. The above result is an identity so it is not
restricted to small speeds only. As $\mathfrak{m}_{k}$ at $0$ and $t$ may have
different probabilities, we take the ensemble average using joint
probabilities to obtain%
\end{subequations}
\begin{subequations}
\begin{equation}
\Delta_{\text{i}}W_{\text{BP}}=(m/2)(\left\langle \mathsf{V}^{2}%
(0)\right\rangle -\left\langle \mathsf{V}^{2}(t)\right\rangle )\geq0.
\label{Irreversible-BP-Energy}%
\end{equation}

From the comments above about the close similarity between the IEQ and EQ
states, we conclude that the velocity distribution is given by the Maxwell
distribution at the instantaneous temperature $T_{\text{CM}}(t)$ of the
degrees of freedom associated with the CM-motion as described in Sec.
\ref{Sec-IEQ}. Thus, we have the conventional result
\begin{equation}
\left\langle \mathsf{V}^{2}(t)\right\rangle =3T_{\text{CM}}(t)/m
\label{Av-Velocity-Square}%
\end{equation}
for the BP in an IEQ state so that
\begin{equation}
\Delta_{\text{i}}W_{\text{BP}}=(3/2)(T_{\text{CM}}(0)-T_{\text{CM}}(t))\geq0,
\label{Irreversible-BP-Energy-Temperature}%
\end{equation}
showing that the slowing down of the CM-motion results in its temperature
falling as time goes on. Eventually, $T_{\text{CM}}(t)\rightarrow T_{0}$. In
EQ, the Brownian motion does not undergo any temperature change ($\Delta
_{\text{i}}W_{\text{BP}}=0$) as is the case for the Langevin equation, even
though $\mathbf{V}_{k}$ varies over $\mathfrak{m}_{k}$.

In the Langevin limit, the equation of motion for $\mathfrak{m}_{k}$ with a
time-dependent $\gamma_{k}(t)$ becomes (see Eq. (\ref{MicroForce-Example}) for
justification)
\end{subequations}
\begin{subequations}
\begin{equation}
d\mathbf{V}_{\text{$k$}}(t)/dt=-(\gamma_{k}(t)/m)\mathbf{V}_{k}(t),
\label{NewLangevinEq-micro}%
\end{equation}
whose solution is%
\begin{equation}
\mathbf{V}_{\text{$k$}}(t)=\mathbf{V}_{\text{$k$}}(0)\exp(-%
{\textstyle\int\nolimits_{0}^{t}}
\gamma_{\text{$k$}}(u)du/m), \label{MicroVelocity}%
\end{equation}
which is independent of $p_{k}$. To be consistent with Eq.
(\ref{Av-Velocity-Square}), $\gamma_{k}(t)$ must not have a fixed sign over
$\left\{  \mathfrak{m}_{k}\right\}  $. This is consistent with the observation
that $\Delta_{\text{i}}W_{k\text{,BP}}$\ has no particular sign over $\left\{
\mathfrak{m}_{k}\right\}  $. This means that the components of the possible
velocities can range from $-\infty$ to $+\infty$ to satisfy Eq.
(\ref{Av-Velocity-Square}). Consequently,%
\end{subequations}
\begin{equation}
\left\langle \mathsf{V}^{2}(t)\right\rangle =\left\langle \mathsf{V}%
^{2}(0)\exp(-2%
{\textstyle\int\nolimits_{0}^{t}}
\mathsf{\gamma}(u)du/m)\right\rangle \leq\left\langle \mathsf{V}%
^{2}(0)\right\rangle , \label{Av-V-Sq}%
\end{equation}
where the last inequality follows from Eq. (\ref{Irreversible-BP-Energy}) and
where the random variable $\mathsf{\gamma}$ has outcomes $\left\{  \gamma
_{k}\right\}  $. We thus see that our approach has allowed the equipartition
theorem to remain valid at all times. From $\left\langle \mathsf{V}%
^{2}(t)\right\rangle \propto1/m$, we conclude that larger the mass, smaller
the mean square fluctuations such as for a macroscopic piston. However, for a
mesoscopic Brownian particle, it can be appreciable and can be observed.

For the Langevin case $\gamma_{\text{$k$}}(u)\simeq\gamma>0$ for all $k$ so
that $\left\langle \mathsf{V}^{2}(t)\right\rangle =e^{-2\gamma t/m}%
\left\langle \mathsf{V}^{2}(0)\right\rangle \leq\left\langle \mathsf{V}%
^{2}(0)\right\rangle $, which is consistent with the above inequality but
shows that $\left\langle \mathsf{V}^{2}(t)\right\rangle \rightarrow0$ as
$t\rightarrow\infty$. This highlights another important difference from the
Langevin approach: in the $\mu$NEQT, $\gamma_{k}(t)$ has \emph{no} sign
restriction. Because of this, it cannot be taken out of the averaging process.
By taking it out in the Langevin case results in an incorrect answer. To see
it clearly, we evaluate $\Delta_{\text{i}}W_{\text{BP}}$ in Eq.
(\ref{Irreversible-BP-Energy}) to obtain $\Delta_{\text{i}}W_{\text{BP}%
}=(m/2)\left\langle \mathsf{V}^{2}(0)\right\rangle (1-e^{-2\gamma t/m})\geq0$.
As $t\rightarrow\infty$, $\Delta_{\text{i}}W_{\text{BP}}=(m/2)\left\langle
\mathsf{V}^{2}(0)\right\rangle >0$, while it must vanish in EQ as noted above.

\subsection{EQ Diffusion}

We now determine the average square displacement of the BP over a long time.
For the sake of simplicity, we will only consider diffusion in an EQ state as
considering IEQ states creates complications that we wish to avoid. We
consider the relative displacement $\Delta\mathbf{R}_{k}\mathbf{(}%
t,0\mathbf{)\doteq R}_{k}(t)-\mathbf{R}_{k}(0)$ over all $\left\{
\mathfrak{m}_{k}\right\}  $ at long time and follow Einstein\ again
\cite{Einstein-BrownianMotion}. The distribution function of the relative
displacement $\Delta\mathbf{R}_{k}=\Delta\mathbf{\mathbf{R}}_{k}(t,0)$ over
all $\left\{  \mathfrak{m}_{k}\right\}  $\ is given by $p_{k}(\Delta
\mathsf{R},t)=e^{-\Delta\mathsf{R}^{2}/4Dt}/(4\pi Dt)^{3/2}$, so that
\begin{equation}
\left\langle \Delta\mathsf{R}^{2}(t)\right\rangle =6Dt
\label{EinsteinRelation}%
\end{equation}
as a function of time; here, $D$ is the diffusion constant, which is related
to the viscosity of the fluid by $D=T/6\pi\eta a$. We can also compute
$\left\langle \Delta\mathsf{R}^{2}(t)\right\rangle $ from\ Eq.
(\ref{NewLangevinEq-micro}) in a standard way but we will not stop to do that.
\newline

\section{Computational Scheme \label{Sec-Example}}

To establish the feasibility of our theory, we describe the computational
methodology by considering a simpler version of the case studied in Sec.
\ref{Sec-GeneralCase}:\ $\Sigma_{\text{BP}}$ as the system $\Sigma$
(previously denoted by $\Sigma_{\text{BP}}$) in a medium $\widetilde{\Sigma}$,
which with $\Sigma$ forms the isolated system $\Sigma_{0}$ (as before). It is
this version that is normally considered in the literature, and its
computational scheme must be consistent with the $\mu$NEQT we have already
developed in the previous sections, except that $\Sigma_{\text{R}}$ is absent
in the current consideration. The BP contains $N_{\text{BP}}$\ particles, each
of mass $m_{\text{BP}}$ so that $M=N_{\text{BP}}m_{\text{BP}}$, and
$\widetilde{\Sigma}$ contains\ $\widetilde{N}$ particles, each of mass
$\widetilde{m}$, so that its total mass is $\widetilde{M}=\widetilde
{N}\widetilde{m}$. In addition, $V$ and $\widetilde{V}$\ are the volumes of
the BP and the medium, respectively, which we keep fixed along with
$N_{\text{BP }}$ and $\widetilde{N}$. Accordingly, we do not exhibit them as
parameters in the Hamiltonians. Our choice for the volumes means that there is
no "pressure-volume" work so the only microwork we need to consider is due to
the microforce $\mathbf{F}_{k,\text{BP}}$ resulting in $d_{\text{i}%
}W_{k\text{,BP}}$; note that now $k$ refers to the BP microstate
$\mathfrak{m}_{k}$. Earlier, $\mathfrak{m}_{k}$ refereed to the joint
macrostates of $\Sigma_{\text{BP}}$ and $\Sigma_{\text{R}}$. This simplifies
the computational complexity considerably.

We will take $\widetilde{\Sigma}$ to be in equilibrium as before in a
canonical ensemble with equilibrium temperature $T_{0}$
\begin{subequations}
\begin{equation}
\widetilde{p}_{\widetilde{k}}=\exp[\beta_{0}(\widetilde{F}-\widetilde
{E}_{\widetilde{k}})] \label{MicroProbability-Medium}%
\end{equation}
for its microstate $\widetilde{\mathfrak{m}}_{\widetilde{k}}$; here,
$\widetilde{F}$ is the thermodynamic potential (the Helmholtz free energy). We
take $\Sigma$, \textit{i.e.}, $\Sigma_{\text{BP}}$ in an internal equilibrium
with temperature $T$ and microforce $\mathbf{F}_{k,\text{BP}}$; its microstate
probability is%
\begin{equation}
p_{k}=\exp[\beta(\Phi-E_{k}+\mathbf{F}_{k,\text{BP}}\mathbf{\cdot R})],
\label{MicroProbability-BP}%
\end{equation}
cf. Eq. (\ref{IEQ-probabilities}), where $\Phi$ is the thermodynamic
potential, and $\mathbf{F}_{k,\text{BP}}$ is given in Eq.
(\ref{MicroForce-Velocity}) and determined below for the current case. As
noted earlier, $\beta_{0},\beta$ and $\mathbf{R}$ are the Lagrange multipliers
so they are not fluctuating over $\widetilde{\mathfrak{m}}_{\widetilde{k}}$
and $\mathfrak{m}_{k}$, as the case may be. In other words, they are
parameters and must be held \emph{fixed }as we use these probabilities to take
averages. However, we can also treat $\mathbf{F}_{\text{BP}}$ as the parameter
so that $\mathbf{R}_{k}$ is fluctuating; cf. Eq. (\ref{IEQ-probabilities0}).
However, we will not do so here.

\subsection{Various Frames}

We consider three different frames $K_{0},\widetilde{K}$, and $K$, in which
$\Sigma_{0},\widetilde{\Sigma}$, and $\Sigma$ are at rest, respectively, with
their CM's at the origins of the frames. These frames determine their
thermodynamic, \textit{i.e., \emph{internal} }energies $E_{0},\widetilde{E}$
and $E$, respectively. The Hamiltonian of $\Sigma_{0}$\ is denoted by
$\mathcal{H}_{0}(\left.  \mathbf{z}_{0}\mathbf{,}\widetilde{\mathbf{z}}%
_{0}\right\vert \mathbf{R}_{\text{BP}},\widetilde{\mathbf{R}})$ with
$\mathbf{z}_{0}=\left\{  \mathbf{z}_{0i}=(\mathbf{x}_{0i},\mathbf{p}%
_{0i})\right\}  _{i=1}^{N_{\text{BP}}}$ and $\widetilde{\mathbf{z}}%
_{0}=\left\{  \widetilde{\mathbf{z}}_{0j}=(\widetilde{\mathbf{x}}%
_{0j},\widetilde{\mathbf{p}}_{0j})\right\}  _{j=1}^{\widetilde{N}}$ referring
to the particles composing the BP\ and $\widetilde{\Sigma}$, respectively, and
$\mathbf{R}_{\text{BP}}$ and $\widetilde{\mathbf{R}}\ $denoting the
CM-displacements of the BP\ and $\widetilde{\Sigma}$, respectively, with
respect to $K_{0}$. In the following, we will exclusively use $i$ for a
BP-particle and $j$ for a medium particle. The Hamiltonian of $\Sigma$ with
respect to $K$ and of $\widetilde{\Sigma}$ with respect to $\widetilde{K}$ are
denoted by $\mathcal{H}(\left.  \mathbf{z}\right\vert \mathbf{R}_{\text{BP}})$
and $\widetilde{\mathcal{H}}(\left.  \widetilde{\mathbf{z}}\right\vert
\widetilde{\mathbf{R}})$, respectively (see below); here $\mathbf{z}=\left\{
\mathbf{z}_{i}=(\mathbf{x}_{i},\mathbf{p}_{i})\right\}  $ and $\widetilde
{\mathbf{z}}=\left\{  \widetilde{\mathbf{z}}_{j}=(\widetilde{\mathbf{x}}%
_{j},\widetilde{\mathbf{p}}_{j})\right\}  $ refer to $K$, and $\widetilde{K}$,
respectively. By definition,%
\end{subequations}
\begin{equation}%
\begin{array}
[c]{c}%
{\textstyle\sum\nolimits_{i=1}^{N_{\text{BP}}}}
\mathbf{z}_{i}=0,%
{\textstyle\sum\nolimits_{j=1}^{\widetilde{N}}}
\widetilde{\mathbf{z}}_{j}=0,\\
m_{\text{BP}}%
{\textstyle\sum\nolimits_{i=1}^{N_{\text{BP}}}}
\mathbf{x}_{0i}+\widetilde{m}%
{\textstyle\sum\nolimits_{j=1}^{\widetilde{N}}}
\widetilde{\mathbf{x}}_{0j}=0,\\%
{\textstyle\sum\nolimits_{i=1}^{N_{\text{BP}}}}
\mathbf{p}_{0i}+%
{\textstyle\sum\nolimits_{j=1}^{\widetilde{N}}}
\widetilde{\mathbf{p}}_{0j}=0.
\end{array}
\label{Frame-Sum-Rule}%
\end{equation}
The first sum in the last equation is $\mathbf{P}_{\text{BP}}$ and the second
sum is $\widetilde{\mathbf{P}}=-\mathbf{P}_{\text{BP}}$.

\subsection{Separating Center of mass Motion}

Let $(\mathbf{R}_{\text{BP}},\mathbf{V}_{\text{BP}})$ and $(\widetilde
{\mathbf{R}},\widetilde{\mathbf{V}})$ denote the (displacement, velocity) of
the CM's of $K$ and $\widetilde{K}$, respectively, with respect to $K_{0}$. We
have $\mathbf{R=R}_{\text{BP}}-\widetilde{\mathbf{R}}$ as the relative
displacement and $\mathbf{V=V}_{\text{BP}}-\widetilde{\mathbf{V}}$ as the
relative velocity of the BP relative to $\widetilde{\Sigma}$. Then%
\begin{align*}
\mathbf{x}_{0i}  &  \doteq\mathbf{x}_{i}+\mathbf{R}_{\text{BP}},\mathbf{p}%
_{0i}\doteq\mathbf{p}_{i}+m_{\text{BP}}\mathbf{V}_{\text{BP}},\\
\widetilde{\mathbf{x}}_{0i}  &  \doteq\widetilde{\mathbf{x}}_{i}%
+\widetilde{\mathbf{R}},\widetilde{\mathbf{p}}_{0i}\doteq\widetilde
{\mathbf{p}}_{i}+\widetilde{m}\widetilde{\mathbf{V}}.
\end{align*}

We can obtain $E_{0}$ in $K_{0}$ by adding to $\widetilde{E}+E$ the CM-kinetic
energies $E_{\text{CM}}\doteq\mathbf{P}_{\text{BP}}^{2}/2M+$ $\widetilde
{\mathbf{P}}^{2}/2\widetilde{M}=\mathbf{P}_{\text{BP}}^{2}/2m$ of the BP and
$\widetilde{\Sigma}$, and their mutual interaction energy $\widehat{U}$:%
\begin{subequations}
\begin{equation}
E_{0}=E_{\text{CM}}+\widetilde{E}+\widehat{U}+E; \label{TotalEnergy}%
\end{equation}
cf. Eq. (\ref{MacroEnergy-Relation}). It should be noted that because of no
pressure-volume work here, this interaction energy is responsible for the
viscous drag. As Eq. (\ref{TotalEnergy}) is a purely mechanical relation, it
also refers to a microstate $\mathfrak{m}_{0k_{0}}$ $(k_{0}=\widetilde
{k}\otimes k)$ of $\Sigma_{0}$ in terms of the microstates $\widetilde
{\mathfrak{m}}_{\widetilde{k}}$ and $\mathfrak{m}_{k}$ of the medium and the
BP, respectively. Thus, $E_{0},\widetilde{E},E$, and $\widehat{U}$ can be
replaced, respectively, by their microanalogs $E_{0k_{0}},\widetilde
{E}_{\widetilde{k}}$ and $E_{k}$, and $\widehat{U}_{k\widetilde{k}}$ or
$\widehat{U}_{k_{0}}$, while we can replace $E_{\text{CM}}$ by its microanalog
$E_{\text{CM,}k}$ or $E_{\text{CM,}\widetilde{k}}$. Thus, we have the identity%
\begin{equation}
E_{0k_{0}}=E_{\text{CM,}k}+E_{k}+\widehat{U}_{k\widetilde{k}}+\widetilde
{E}_{\widetilde{k}}; \label{Microenergy-Isolated}%
\end{equation}
cf. Eq. (\ref{MicroEnergy-Relation}). We obtain various macroenergies from
them:
\end{subequations}
\begin{subequations}
\begin{align}
E_{0}  &  =%
{\textstyle\sum\nolimits_{k_{0}}}
p_{k_{0}}E_{0k_{0}},\widehat{U}=%
{\textstyle\sum\nolimits_{k_{0}}}
p_{k_{0}}\widehat{U}_{k\widetilde{k}},\nonumber\\
\widetilde{E}  &  =%
{\textstyle\sum\nolimits_{\widetilde{k}}}
p_{\widetilde{k}}\widetilde{E}_{\widetilde{k}}=%
{\textstyle\sum\nolimits_{k_{0}}}
p_{k_{0}}\widetilde{E}_{\widetilde{k}},\label{Microenergies}\\
E  &  =%
{\textstyle\sum\nolimits_{k}}
p_{k}E_{k}=%
{\textstyle\sum\nolimits_{k_{0}}}
p_{k_{0}}E_{k}.\nonumber
\end{align}
Let us introduce conditional probabilities $p(\widetilde{k}\mid k)$ of
$\widetilde{k}$ given $k$ so that $p_{k_{0}}=p_{k}p(\widetilde{k}\mid k)$. We
use them to determine microenergies $E_{0k}$ and $\widetilde{E}_{k}$ of
$\Sigma_{0}$\ and $\widetilde{\Sigma}$, respectively,\ that can be associated
with $\mathfrak{m}_{k}$:%
\begin{align}
E_{0k}  &  \doteq%
{\textstyle\sum\nolimits_{\widetilde{k}}}
p(\widetilde{k}\mid k)E_{k_{0}},\nonumber\\
\widetilde{E}_{k}  &  \doteq%
{\textstyle\sum\nolimits_{\widetilde{k}}}
p(\widetilde{k}\mid k)\widetilde{E}_{\widetilde{k}}, \label{RestrictedAverage}%
\\
E_{0}  &  =%
{\textstyle\sum\nolimits_{k}}
p_{k}E_{0k},\widetilde{E}=%
{\textstyle\sum\nolimits_{k}}
p_{k}\widetilde{E}_{k}.\nonumber
\end{align}
If the medium and the BP are quasi-independent, then $p(\widetilde{k}\mid
k)=p_{\widetilde{k}}$.

As $\widetilde{\Sigma}$ is in equilibrium, no irreversibility can be
associated with it. Accordingly, we ascribe the irreversibility to the BP
itself in the following. For this, we need to obtain the interaction energy
associated with the BP alone so we must average $\widehat{U}_{k,\widetilde{k}%
}$ over the medium (see the second equation in Eq. (\ref{RestrictedAverage}))
as below:
\begin{equation}
\widehat{U}_{k}\doteq%
{\textstyle\sum\nolimits_{\widetilde{k}}}
p(\widetilde{k}\mid k)\widehat{U}_{k,\widetilde{k}}; \label{AdhesiveEnergy}%
\end{equation}
it is the interaction energy associated with $\mathfrak{m}_{k}$ for a given
macrostate of the medium. We can now extract $E_{0k}$ from Eq.
(\ref{Microenergy-Isolated}):%
\end{subequations}
\begin{equation}
E_{0k}=E_{\text{CM,}k}+E_{k}+\widehat{U}_{k}+\widetilde{E}_{k}. \label{E_0k}%
\end{equation}

We did not explicitly consider $\widehat{U}$ in Secs. \ref{Sec-MicroTH} and
\ref{Sec-ViscousDrag} as we were considering the microstates of $\Sigma
_{\text{BP}}$ and $\Sigma_{\text{R}}$\ together forming $\Sigma$. This
tremendously simplified our discussion there as $\widehat{U}_{k}$ was already
included in $E_{k}$. The situation is different here where we consider
$\Sigma_{\text{BP}}$ and $\widetilde{\Sigma}$ together but we wish to consider
the microstate $\mathfrak{m}_{k}$ of the BP alone so we need to deal with
$\widehat{U}_{k,\widetilde{k}}$ separately and determine $\widehat{U}_{k}$ in
our discussion.

\subsection{Hamiltonians}

We first focus on the BP. In its rest frame $K$, its Hamiltonian is given by%
\begin{subequations}
\begin{equation}
\mathcal{H}_{\text{BP}}(\mathbf{z})=%
{\textstyle\sum\nolimits_{i=1}^{N_{\text{BP}}}}
\mathbf{p}_{i}^{2}/2m_{\text{BP}}+%
{\textstyle\sum\nolimits_{i,i^{\prime}=1}^{\prime N_{\text{BP}}}}
U(\mathbf{x}_{i}-\mathbf{x}_{i^{\prime}}), \label{BP-Hamiltonian-Ex}%
\end{equation}
where $U(\mathbf{x}_{i}-\mathbf{x}_{i^{\prime}})$ is the potential energy
between BP particles at $\mathbf{x}_{i}$ and $\mathbf{x}_{i^{\prime}}$. The
prime over the summation implies $i\neq i^{\prime}$. The Hamiltonian when
applied to $\mathfrak{m}_{k}$ determines $E_{k}$ so that $E=\left\langle
\mathcal{H}_{\text{BP}}\right\rangle =\left\langle \mathsf{E}\right\rangle $
with microstate probabilities $p_{k}$ given above.

The Hamiltonian of $\widetilde{\Sigma}$ in its rest frame $\widetilde{K}$ is
given by%
\begin{equation}
\widetilde{\mathcal{H}}(\widetilde{\mathbf{z}})=%
{\textstyle\sum\nolimits_{j=1}^{\widetilde{N}}}
\widetilde{\mathbf{p}}_{j}^{2}/2\widetilde{m}+%
{\textstyle\sum\nolimits_{j,j^{\prime}=1}^{\prime\widetilde{N}}}
\widetilde{U}(\widetilde{\mathbf{x}}_{j}-\widetilde{\mathbf{x}}_{j^{\prime}}),
\label{Medium-Hamiltonian}%
\end{equation}
where $\widetilde{U}(\widetilde{\mathbf{x}}_{j}-\widetilde{\mathbf{x}%
}_{j^{\prime}})$ is the potential energy between the particles of
$\widetilde{\Sigma}$ at $\widetilde{\mathbf{x}}_{j}$ and $\widetilde
{\mathbf{x}}_{j^{\prime}}$. It gives the microstate energy $\widetilde
{E}_{\widetilde{k}}$ of $\widetilde{\mathfrak{m}}_{\widetilde{k}}$ so that
$\widetilde{E}=\left\langle \widetilde{\mathcal{H}}\right\rangle
_{\widetilde{\Sigma}}=\left\langle \widetilde{\mathsf{E}}\right\rangle
_{\widetilde{\Sigma}}$ with microstate probabilities $\widetilde
{p}_{\widetilde{k}}$ given above.

In terms of these Hamiltonians, $\mathcal{H}_{0}$ is given by%
\begin{align}
\mathcal{H}_{0}  &  =\mathcal{H}_{\text{BP}}(\left.  \mathbf{z}_{0}\right\vert
\mathbf{R}_{\text{BP}})+\widetilde{\mathcal{H}}(\left.  \widetilde{\mathbf{z}%
}_{0}\right\vert \widetilde{\mathbf{R}})\nonumber\\
&  +%
{\textstyle\sum\nolimits_{i=1}^{N_{\text{BP}}}}
{\textstyle\sum\nolimits_{j=1}^{\widetilde{N}}}
\widehat{U}(\mathbf{x}_{0i}-\widetilde{\mathbf{x}}_{0j}),
\label{IsolatedSys-Hamiltonian}%
\end{align}
in the $K_{0}$ frame. Here, $\widehat{U}(\mathbf{x}_{0i}-\widetilde
{\mathbf{x}}_{0j})$ denotes the mutual potential energy between a BP particle
at $\mathbf{x}_{0i}$ and a medium particle at $\widetilde{\mathbf{x}}_{0j}$.
The mutual interaction between the BP and $\widetilde{\Sigma}$ is described by
the last term above.

\subsection{Frame Change}

We need to express the interaction energy explicitly in terms of $\mathbf{R}$;
cf. Eqs. (\ref{V-F_Definition}) and (\ref{MicroForce-Velocity}). This will be
needed below; see Eq. (\ref{MicroForce-Example}). Using the identity
$\mathbf{x}_{0i}-\widetilde{\mathbf{x}}_{0j}\mathbf{=x}_{i}-\widetilde
{\mathbf{x}}_{j}+\mathbf{R}$, we rewrite $\widehat{U}(\mathbf{x}%
_{0i}-\widetilde{\mathbf{x}}_{0j})=\widehat{U}(\mathbf{x}_{i}-\widetilde
{\mathbf{x}}_{j}+\mathbf{R})$; here, $\mathbf{x}_{i}$ is defined in the $K$
frame and $\widetilde{\mathbf{x}}_{j}$ is defined in the $\widetilde{K}$ frame
so that we can manipulate this energy conveniently as required below. The last
sum in Eq. (\ref{IsolatedSys-Hamiltonian}) is the potential energy between the
BP in a given\ microstate $\mathfrak{m}_{k}$ and the medium in a given
microstate $\widetilde{\mathfrak{m}}_{\widetilde{k}}$, and defines
$\widehat{U}_{k\widetilde{k}}$. To obtain $\widehat{U}_{k}$, we need to
average it according to Eq. (\ref{AdhesiveEnergy}) using $p(\left.
\widetilde{k}\right\vert k)$. While this can be done, for computational
simplicity here, we will average using $\widetilde{p}_{\widetilde{k}}$, the
canonical distribution of $\widetilde{\Sigma}$ in its rest frame, given above,
which we denote by $\left\langle {}\right\rangle _{\widetilde{\Sigma}}$. It is
a \emph{restricted} average and gives us%
\end{subequations}
\begin{equation}
\widehat{U}_{k}(T_{0},\mathbf{R})\doteq%
{\textstyle\sum\nolimits_{i=1}^{N_{\text{BP}}}}
\left\langle
{\textstyle\sum\nolimits_{j=1}^{\widetilde{N}}}
\widehat{U}(\mathbf{x}_{i}+\mathbf{R}-\widetilde{\mathbf{x}}_{j})\right\rangle
_{\widetilde{\Sigma}}, \label{MicroPotential}%
\end{equation}
for $\mathfrak{m}_{k}$ determined by $\left\{  \mathbf{x}_{i}\right\}  $; see
the definition of $\widetilde{E}_{k}$ in Eq. (\ref{RestrictedAverage}); the
dependence on $T_{0}$ is due to the above averaging.

\subsection{Determination of $\mathbf{F}_{\text{BP}}$ and $\mathbf{F}%
_{k,\text{BP}}$\label{Sec-Macro-MicroForces}}

The average of $\widehat{U}_{k}(T_{0},\mathbf{R})$ over $p_{k}$ determines the
macroscopic potential
\begin{equation}
\widehat{U}(T,T_{0},u)=\left\langle \widehat{\mathsf{U}}(T_{0},\mathbf{R}%
)\right\rangle , \label{Medium-Averaging-U}%
\end{equation}
where $\widehat{\mathsf{U}}(T_{0},\mathbf{R})$ is the random variable with
outcomes $\left\{  \widehat{U}_{k}(T_{0},\mathbf{R})\right\}  $, $T$ appears
due to the ensemble averaging, and $u$ is defined in Eq.
(\ref{ScalarCombination}) as we now explain. This average is the $\widehat{U}$
in $E_{0}$ in Eq. (\ref{TotalEnergy}). Using $E_{0}$ for $E$ in Eq.
(\ref{V-F_Definition}) and recognizing that only $\widehat{U}$ depends on
$\mathbf{R}$, we realize that this dependence must be through $u$ as explained
in Sec. \ref{Sec-ViscousDrag}. We thus find that the macroforce $\mathbf{F}%
_{\text{BP}}$ is given by%
\begin{equation}
\mathbf{F}_{\text{BP}}(t)\mathbf{=V}\frac{\partial\widehat{U}(T,T_{0}%
,u)}{\partial u},\frac{\partial\widehat{U}(T,T_{0},u)}{\partial u}<0,
\label{MacroForce-Example}%
\end{equation}
in which $\partial\widehat{U}(T,T_{0},u)/\partial u$ represent
$(-f(T,\mathbf{V},t))$ in Eq. (\ref{FrictionForm}) for the current case; cf.
Eq. (\ref{General-friction-f}).

We now use Eq. (\ref{MicroForce-Velocity}) by replacing $E_{k}$ there with
$E_{0k}$. We obtain%
\begin{equation}
\mathbf{F}_{k,\text{BP}}(t)=\frac{\partial\widehat{U}_{k}(T_{0},\mathbf{R}%
)}{\partial\mathbf{R}} \label{MicroForce-Example}%
\end{equation}
This determines the microforce thermodynamically.

As we have already discussed the CM-motion of the BP, we can go back to the
$K$ frame and consider $\mathcal{H}_{\text{BP}}(\mathbf{z})$ to write down the
equations of motions for $\mathbf{x}_{i}$ and $\mathbf{p}_{i}$ for
$\mathfrak{m}_{k}$,%
\begin{equation}
\frac{d\mathbf{p}_{i}}{dt}=-%
{\textstyle\sum\nolimits_{i^{\prime}}^{\prime N_{\text{BP}}}}
\frac{\partial U(\mathbf{x}_{i}\mathbf{-x}_{i^{\prime}})}{\partial
\mathbf{x}_{i}},m_{\text{BP}}\frac{d\mathbf{x}_{i}}{dt}=\mathbf{p}%
_{i},\label{InternalMotionBP}%
\end{equation}
dealing only with internal forces. These equations determine how
$\mathfrak{m}_{k}$ evolves in time and determine the evolution of $\left\{
\mathbf{z}_{i}\right\}  $ in time in the $K$ frame. From this, we extract
$\left\{  \mathbf{x}_{i}(t)\right\}  $ to be used in Eq.
(\ref{Medium-Averaging-U}) to determine $\widehat{U}_{k}(T_{0},\mathbf{R})$.
We then determine $\mathbf{F}_{k,\text{BP}}(t)$ and follow the prescription of
the solution of Eq. (\ref{NewLangevinEq}) over $\Delta t$.

Note that we do not need to solve the Hamilton's equations for the particles
in $\widetilde{\Sigma}$, which provides a major simplification of our
approach. The stochasticity, as we have mentioned several times, emerges and
is completely captured when we average over $\mathfrak{m}_{k}$ using $p_{k}$
from Eq. (\ref{MicroProbability-BP}).

The BP equation of motion in Eq. (\ref{NewLangevinEq}) differs from the
original Langevin equation in Eq. (\ref{LangevinEquation0}) in that it is
missing the partitioning shown in Eq. (\ref{LangevinEquation}). Since the
Hamiltonian in Eq. (\ref{IsolatedSys-Hamiltonian}) does not have any
stochasticity, $\mathbf{F}_{k,\text{BP}}(t)$ above cannot be compared with the
$\mathbf{F}_{k,\text{BP}}^{\prime}(t)$ in Eq. (\ref{LangevinEquation0}). It
follows from Eq. (\ref{NewLangevinEq} ) that $\mathbf{V=V}_{k}$ in a given
BP-microstate $\mathfrak{m}_{k}$\ is a slowy-varying function that is
differentiable. The random fluctuations in it are described by considering it
over microstates. Thus, the fast fluctuations similar but not identical to
$\boldsymbol{\xi}$ are captured when we consider the ensemble $\left\{
\mathbf{V}_{k}\right\}  $ or $\left\{  \mathbf{F}_{k,\text{BP}}\right\}
$\ over $\left\{  \mathfrak{m}_{k}\right\}  $. The fluctuations in the medium
are not relevant in our approach as we have performed an average over
$\widetilde{\mathfrak{m}}_{\widetilde{k}}$ for reasons explained above.

We can use well-established solution and surface thermodynamic theories
\cite{Israelachvili} to determine useful forms of $\widehat{U}$, which can be
useful to validate any approximation, if any, is made in the evaluation of
$\widehat{U}_{k}$. In most cases, the interaction potential can be
approximated by considering $N_{\text{BP}}^{\text{(s)}}$ and $\widetilde
{N}^{\text{(s)}}$\ particles in a thin interface surrounding the BP-surface
and not all the particles. This will provide a further simplification in the
calculation. The averaging in Eq. (\ref{Medium-Averaging-U}) is limited to
only the thin interface and not the entire medium. Indeed, Einstein used this
interface to determine the osmotic force in his analysis of the Brownian
motion. Other sophisticated techniques may also be useful to deal with the
computation \cite{Frenkel,Woodcock,Bussi,Andersen,Nose}.

\subsection{A BP with $N_{\text{BP}}=1$\label{Sec-BP-N=1}}

Let us consider the simplest possible case $N_{\text{BP}}=1$ that is commonly
studied; see for example \cite{Kapral0}. In this case, Eq.
(\ref{InternalMotionBP}) is meaningless and $\widehat{U}_{k}$ depends only on
a single phase point $\mathbf{z}=0$ of the BP, which also refers to its CM.
The BP-microstate $\mathfrak{m}_{k}$\ is now the small cell $\delta\mathbf{z}$
around $\mathbf{z}$ in a $6$-dimensional phase space. With respect to the CM
of $\widetilde{\Sigma}$, the BP has a relative displacement $\mathbf{R}$ and a
relative velocity $\mathbf{V}_{k}$. We determine $\widehat{U}_{k}%
(T_{0},\mathbf{R})$ for general $\mathbf{R}$. This has to be done once. We now
determine $\mathbf{F}_{k,\text{BP}}(t)$ using Eq. (\ref{MicroForce-Example})
and $\mathbf{F}_{\text{BP}}(t)$\ for all possible values of $\mathbf{R}$. We
then integrate the equation of motion using basic calculus.

As an example, let us assume that $\widehat{U}_{k}(T_{0},\mathbf{R})$ is given
by
\begin{equation}
\widehat{U}_{k}(T_{0},\mathbf{R})=-\gamma_{k}(T_{0},t)\mathbf{R}%
(t)\cdot\mathbf{V}_{k}(t),\label{N_BP=1-U-Ex}%
\end{equation}
in terms of the scalar product; cf. Sec. \ref{Sec-ViscousDrag}. This results
in
\begin{subequations}
\begin{equation}
\mathbf{F}_{k,\text{BP}}(t)=-\gamma_{k}\mathbf{V}_{k}(t),d_{\text{i}%
}W_{k,\text{BP}}(t)=\gamma_{k}\mathbf{V}_{k}^{2}(t)dt.\label{N_BP=1-F-Ex}%
\end{equation}
Comparing with Eq. (\ref{NewLangevinEq-micro}), we see that $\mathbf{V}%
_{k}(t)$ after basic integration is given by Eq. (\ref{MicroVelocity}). Note
again that we do not need to consider $\widetilde{\Sigma}$ in obtaining the
solution, which demonstrate the usefulness of the new theory. The macroforce
$\mathbf{F}_{\text{BP}}(t)=\left\langle -\mathsf{\gamma}_{k}(T_{0}%
,t)\mathsf{V}_{k}(t)\right\rangle $ must follow the form for viscous drag in
Eq. (\ref{MacroForce-Example}) so we have%
\begin{equation}
\mathbf{F}_{\text{BP}}(t)=-\gamma_{\text{eff}}\mathbf{V}(t),d_{\text{i}%
}W_{\text{BP}}(t)=\gamma_{\text{eff}}\mathbf{V}^{2}(t)dt,\gamma_{\text{eff}%
}>0,\label{N_BP=1-F-eff-Ex}%
\end{equation}
where $d\mathbf{R}(t)/dt=\mathbf{V}(t)\doteq\left\langle \mathsf{V}%
\right\rangle $, and $\gamma_{\text{eff}}$ is an effective parameter defined
by $\gamma_{\text{eff}}(T_{0},t)\mathbf{V}^{2}(t)\doteq\left\langle
\mathsf{\gamma}(T_{0},t)\mathsf{V}^{2}(t)\right\rangle $.

\subsection{Returning to the BP in $\Sigma$}

We now return to the earlier case of the BP as a part of $\Sigma$. In this
case, $\widehat{U}_{k}(\mathbf{R})\ $is the analog of $\widehat{U}_{k}%
(T_{0},\mathbf{R})$ for the microstate $\mathfrak{m}_{k}$ of $\Sigma$. Its
ensemble average $\widehat{U}(T,u)\doteq\left\langle \widehat{U}%
_{k}(\mathbf{R})\right\rangle $ over $\mathfrak{m}_{k}$ is subsumed into $E$
in Eq. (\ref{V-F_Definition}); the dependence on $\mathbf{R}$ appears through
$u$ as explained in Sec. \ref{Sec-ViscousDrag}. Because of this, there is no
reason to extract it from $E$ so we identify the viscous force $\mathbf{F}%
_{\text{BP}}(t)$ by differentiating $E$ with respect to $u$; see Eq.
(\ref{General-friction-f}). This macroforce is for the entire system $\Sigma$
and not for just $\Sigma_{\text{BP}}$. Thus, the microforce $\mathbf{F}%
_{k,\text{BP}}(t)$ is for $\Sigma$'s microstate $\mathfrak{m}_{k}$. While we
did not do, it is possible to extract the microforce associated with a
BP-microstate by the method presented in this section.

\section{Discussion and Conclusions}

The present work was motivated by a desire to obtain a deterministic equation
of motion of a BP in microstate $\mathfrak{m}_{k}$ by considering a microstate
thermodynamics ($\mu$NEQT) in order to provide an alternative to the
stochastic Langevin approach that contains the original phenomenological
Langevin equation of motion and the more advanced generalized Langevin
equations such as due to Zwanzig and Mori \cite{Evans-book,Zwanzig}. The
central concept in the latter approaches is the partition of the microforce
acting on a BP into fast and slow components, which finds its formal
justification in the Mori-Zwanzig approach \cite{Evans-book,Zwanzig}. The
stochasticity due to the fast component is determined by the conditional
probability of $\boldsymbol{\xi}$ and the stochasticity of the slow component
is determined by the initial conditions.

In contrast, we have adopted a hybrid approach to derive the equation of
motion of the BP, in the spirit of Langevin, by following not his mechanical
approach but a thermodynamic approach based on the energy $E$, which
generalizes the one adopted by Einstein \cite{Einstein-BrownianMotion}.
Instead of focussing only on the BP and its diffusion in a medium (which we
treat in Sec. \ref{Sec-Example}), we take a comprehensive first-principle
approach to consider the BP as a part of a system $\Sigma$, embedded in a
medium $\widetilde{\Sigma}$; see Fig. \ref{Fig_System}. All of them form the
isolated system $\Sigma_{0}$. We consider\ the rest frame $K_{0}$ in which
$\Sigma_{0},\widetilde{\Sigma}$, and $\Sigma$ are at rest (except in Sec.
\ref{Sec-Example}), but we allow the BP and $\Sigma_{\text{R}}$\ (replaced by
$\widetilde{\Sigma}$ in Sec. \ref{Sec-Example}) to have a relative motion
specified by $\mathbf{R}$\ and $\mathbf{V}$, necessary to describe the process
of viscous drag as we find that the \emph{relative} motion between
$\Sigma_{\text{BP}}$ and $\Sigma_{\text{R}}$ (or $\widetilde{\Sigma}$) is the
source of viscous dissipation. We treat $\mathbf{R}$ as a parameter and
consider the energy $E(S,V,\mathbf{R})$, from which we determine the viscous
force $\mathbf{F}_{\text{BP}}$ as it opposes motion in accordance with the
second law. After identifying this macroforce in our thermodynamic approach,
we go a step further and obtain a deterministic equation of motion of the
microstate $\mathfrak{m}_{k}$\ of $\Sigma$, which was \ not the focus of
Einstein.\ Later in Sec. \ref{Sec-Example}, we return to the simple system of
a BP in a medium, and obtain the deterministic equation of motion of the
BP-microstate $\mathfrak{m}_{k}$ (not to be confused with $\mathfrak{m}_{k}%
$\ of $\Sigma$ discussed above). There we establish that the mutual
interaction energy $\widehat{U}$ between them is the source of viscous dissipation.

We accomplish our goal by developing a $\mu$NEQT that deals with each
microstate individually. At this level, the microworks are done at fixed
microstate probabilities $p_{k}$ so evaluating them is simplified. The
inherent determinism comes from the fact that the Hamiltonians and the
Hamiltonian equations have no randomness and apply directly to the microstates
individually. At this level, the potential energies in the Hamiltonian
determine various microforces including the one ($\mathbf{F}_{k,\text{BP}}$),
see Eq. (\ref{MicroForce-Example}), responsible for viscous dissipation in
terms of the macroforce $\mathbf{F}_{\text{BP}}$. We do not partition
$\mathbf{F}_{k,\text{BP}}$ into slow and fast components as required in the
Langevin approach. This is one of the distinctions between the two approaches. \ \ \ \ \ \ \ \ \ \ \ \ \ \ \ \ \ 

The stochasticity and the second law emerges automatically in our approach
when the ensemble average over microstates is taken as is standard in
statistical thermodynamics. To obtain the $\mu$NEQT, we need to uniquely
extract from the MNEQT a description suitable at the microstate level. We have
introduced the $\mu$NEQT a while back
\cite{Gujrati-II,Gujrati-Heat-Work0,Gujrati-Heat-Work,Gujrati-Entropy2}. It is
a first-principles theory and its main purpose here is to study BP in NEQ
situations, where the Langevin equation in Eqs. (\ref{LangevinEquation0}) and
(\ref{LangevinEquation}) are inapplicable. In order to have a well-defined NEQ
temperature $T$ of the system, we need to assume the system to be in an IEQ
state requiring internal variables; cf. Sec. \ref{Sec-IEQ}. For simplicity, we
have considered a single internal variable $\mathbf{P}_{\text{BP}}$ or
$\mathbf{R}$ in this study, which along with $V$ are independent state
variables in the state space $\mathfrak{S}^{\prime}$.

The deterministic equation of motion in Eq. (\ref{NewLangevinEq}) for the
microforce is easy to solve as described in Sec. \ref{Sec-Soln-Eq-Motion}. We
do not need the sophisticated concepts like the Wiener process, It\^{o} and
Stratonovich\ integrals, etc. This is a benefit of adopting the $\mu$NEQT. The
method of solution does not require knowing any interaction with
$\widetilde{\Sigma}$; it only requires interactions between $\Sigma
_{\text{BP}}$ and $\Sigma_{\text{R}}$. This becomes very important in Sec.
\ref{Sec-Example} that we will discuss below. As the second law is inoperative
at the microstate level, the microforce $\mathbf{F}_{k,\text{BP}}$ does not
always oppose motion in our approach; that holds only for the macroforce
$\mathbf{F}_{\text{BP}}$. The uniquely defined microforces $\left\{
\mathbf{F}_{k,\text{BP}}\right\}  $ and microworks $\left\{  d_{\text{i}%
}W_{k,\text{BP}}=-\mathbf{F}_{k,\text{BP}}\centerdot d\mathbf{R}\right\}  $
done by them become, as expected, fluctuating quantities over $\left\{
\mathfrak{m}_{k}\right\}  $. We make no assumptions about the nature of these
fluctuations as is needed for the stochastic forces in the Langevin approach.
The internal microwork $d_{\text{i}}W_{k,\text{BP}}\ $also has no fixed sign.
However, $d_{\text{i}}W_{k,\text{BP}}$ and the change in the kinetic energy
$E_{k,\text{CM}}=m\mathbf{V}_{k}^{2}/2$ satisfy \ \ \ \ \ \ \ \ \ \
\end{subequations}
\[
d_{\text{i}}W_{k,\text{BP}}+dE_{k,\text{CM}}=0,
\]
as seen from Eq. (\ref{V-F_BP-Relation}); see also Eq.
(\ref{MicroWork-KineticEnergy2}). This is expected at the microstate level
where classical mechanics operates. Some microforces increase $E_{k,\text{CM}%
}$; some decrease it. Things change at the macroscopic level after ensemble
averaging as seen from Eq. (\ref{Irreversible-BP-Energy}). Now, $d_{\text{i}%
}W_{\text{BP}}\geq0$ is constrained by the second law so the macroforce always
opposes motion.

Within the MNEQT, we determine fluctuations in various thermodynamically
relevant random variables. We limit ourselves to only second order in
expansion so we have only Gaussian fluctuations to reproduce all known EQ
results such as the Einstein relation in Eq. (\ref{EinsteinRelation}). We need
to go to higher order in expansion to obtain non-Gaussian fluctuations, but
the machinery is there. We show that the equipartition theorem is satisfied at
all times as the system is in IEQ, except that the degrees of freedom
associated with the CM-motion have their own temperature $T_{\text{CM}}$,
which may be different from $T$ or $T_{0}$; see the discussion in Sec.
\ref{Sec-IEQ}. This temperature always decreases as the CM-motion ceases as
the system (or the BP) comes to equilibrium; cf. Eq.
(\ref{Irreversible-BP-Energy-Temperature}). Therefore, it cannot be $T$, which
can either go up, down or remain unchanged depending on how it relates to
$T_{0}$. Even if $T=T_{0}$, $T_{\text{CM}}$ will continue to decrease if
$\mathbf{P}_{\text{BP}}\neq0$. We also obtain other results such as a complex
velocity-dependent microscopic friction coefficient $\gamma_{k}$ or the
internal microwork $\Delta_{\text{i}}W_{k,\text{BP}}$, both of which can be of
either sign.

An important aspect of the $\mu$NEQT approach should be mentioned. As
$\mathbf{F}_{k,\text{BP}}$ is oblivious to $p_{k}$, it does not change whether
we are dealing with an EQ case or a NEQ case; the latter are specified only by
$p_{k}$'s. Thus, we can determine $\mathbf{F}_{k,\text{BP}}$ in an EQ
situation, but use it in a NEQ situation by merely using the NEQ $p_{k}$'s. In
EQ, $\mathbf{F}_{\text{BP}}=0$, while in a NEQ case, $\mathbf{F}_{\text{BP}%
}\neq0$. But the fluctuations in $\left\{  \mathbf{F}_{k,\text{BP}}\right\}
$\ are present in both cases.\ The same discussion also applies to $\left\{
d_{\text{i}}W_{k,\text{BP}}\right\}  $.\ \ \ \ \ \ \ 

We have taken $V$ to be some generic work parameter and that the irreversible
macroworks $d_{\text{i}}W_{V}=(P-P_{0})dV$ and $d_{\text{i}}W_{\text{BP}}$
were treated as independent, which allowed us to get the two inequalities in
Eq. (\ref{SecondLawConsequences}). For a genuine piston problem in which the
volume $V$ changes due to piston displacement $dX_{\text{BP}}$, they are not
independent. Indeed, the force imbalance $P-P_{0}$ causes the friction force,
which eventually ensures $P=P_{0}$ in EQ as is well known; see for example
Ref. \cite{Gujrati-I,Gujrati-II}. Recognizing that $dV=AdX_{\text{BP}%
}=AmdX/M_{\text{BP}}$, where $A$ is the area of the piston, we must have
$F_{\text{BP}}=-(P-P_{0})Am/M_{\text{BP}}$ in the $1$-d case. We thus see that
the standard piston behaves as a BP undergoing Brownian motion. A similar
discussion can be carried out for the particle in Fig. (b), which also
undergoes Brownian motion.

The choice of using only IEQ states should not be taken as a limitation of the
$\mu$NEQT. As the system gets farther and farther away from the EQ state, we
need more and more of the independent internal variables, which requires a
larger and larger state space $\mathfrak{S}$ in which the IEQ states are
defined \cite{Gujrati-II,Gujrati-Hierarchy}. This requires a trivial extension
of the present approach. We may also need to treat the system as inhomogeneous
(see \cite{Gujrati-II} for details). Thus, the $\mu$NEQT is capable of
describing any complex NEQ state. The challenge is to identify additional
internal variables.

We have discussed the feasibility of the new theory in Sec. \ref{Sec-Example}
by considering a simple version of the problem often studied in the
literature: a single BP consisting of $N_{\text{BP}}$ particles in a medium
consisting of $\widetilde{N}$ particles. We treat a NEQ macrostate of the BP
by having a temperature difference between the BP and the medium and a
relative motion between them. Here, we relate the microforce $\mathbf{F}%
_{k,\text{BP}}$ to the mutual interaction $\widehat{U}_{k}$ associated with
the BP-microstate $\mathfrak{m}_{k}$. The discussion is very general and
$\widehat{U}_{k}$ includes the mutual interaction between all pairs of medium
and BP particles, except that we average it over all microstates of the
medium; see Eq. (\ref{Medium-Averaging-U}). Thus, $\widehat{U}_{k}$ depends
only on $N_{\text{BP}}$ as it only depends on the microstate $\mathfrak{m}%
_{k}$ of the BP alone. Once this "average" potential is obtained as a function
of $\mathbf{R}$, we do not need to worry about the particles of $\widetilde
{\Sigma}$. We can use $\widehat{U}_{k}$ for $\mathfrak{m}_{k}$, regardless of
whether the BP is in EQ or not with respect to the medium; the latter is
controlled by $p_{k}$. This is the same conclusions we had arrived at for
$\left\{  \mathbf{F}_{k,\text{BP}}\right\}  $\ and $\left\{  d_{\text{i}%
}W_{k,\text{BP}}\right\}  $ above.

A good approximation for $\widehat{U}_{k}$ will be obtained by limiting the
mutual interactions between particles in a thin interface between the BP and
the medium. This is the approximation used by Einstein who used the osmotic
pressure across it to develop his theory. In Sec. \ref{Sec-BP-N=1}, we
consider the case of a BP made up of a single particle ($N_{\text{BP}}=1$) in
a medium so $\mathfrak{m}_{k}$ refers to a single particle in a $6$%
-dimensional phase space. Here, the simplicity of our approach becomes
obvious. Once $\widehat{U}_{k}$ has been obtained, the solution of the
equation of motion requires only following the single particle. A simple model
given in Eq. (\ref{N_BP=1-U-Ex}) clarifies this point.

At a fundamental level, there are subtle but profound differences in the $\mu
$NEQT approach and the Langevin approach. It is important to draw attention to
them before closing, which we list below.

\begin{enumerate}
\item $\mathbf{F}_{k,\text{BP}}$ and $d_{\text{i}}W_{k,\text{BP}}$ are
uniquely determined by $\mathfrak{m}_{k}$, deterministic, and independent of
$p_{k}$. Because of the presence of $\boldsymbol{\xi}$ in the Langevin
approach, $\mathbf{F}_{k,\text{BP}}^{\prime}$ in Eq. (\ref{LangevinEquation0})
and the microwork done by it are random quantities.

\item The trajectory from Eq. (\ref{NewLangevinEq}), being deterministic,
requires integration using basic calculus. The trajectory from the Langevin
equation, being stochastic, requires technical concepts of the Wiener
processes (the It\^{o} and Stratonovich\ integrals), which are not as easy as
the basic calculus.

\item As the sign of $\gamma_{k}$ is not fixed, it cannot be taken out of
averaging in Eq. (\ref{Av-V-Sq}) in the $\mu$NEQT. Doing so in the Langevin
limit gives an unphysical result showing that the fluctuating sign is crucial
for correct physics.

\item The MNEQT approach provides a thermodynamic justification for the
frictional drag for small relative velocities. In the Langevin approach, it
appears phenomenologically. The $\mu$NEQT further unravels the mystery behind
the microforce as noted above.
\end{enumerate}

Thus, we hope that the $\mu$NEQT presented here will prove useful to study
both passive and active BPs, and NEQ\ BP in general.


\begin{thebibliography}{99}                                                                                               %


\bibitem {Landau}L.D.\ Landau, E.M. Lifshitz, \textit{Statistical Physics},
Vol. 1, Third Edition, Pergamon Press, Oxford (1986).

\bibitem {Note}A uniform system in equilibrium cannot sustain relative motions
of its parts \cite{Landau}.

\bibitem {deGroot}S.R. de Groot and P. Mazur, \textit{Nonequilibrium
Thermodynamics}\textbf{, }First Edition, Dover, New York (1984).

\bibitem {Prigogine}D. Kondepudi and I. Prigogine, \textit{Modern
Thermodynamics}, John Wiley and Sons, West Sussex (1998).

\bibitem {Maugin}G.A. Maugin, \textit{The Thermomechanics of Nonlinear
Irreversible Behaviors:\ An Introduction}, World Scientific, Singapore (1999).

\bibitem {Coleman}B.D. Coleman, J. Chem. Phys. \textbf{47}, 597 (1967).

\bibitem {Gujrati-I}P.D. Gujrati, Phys. Rev. E \textbf{81}, 051130 (2010);
P.D. Gujrati, arXiv:0910.0026.

\bibitem {Gujrati-II}P.D. Gujrati, Phys. Rev. E \textbf{85}, 041128 (2012);
P.D. Gujrati, arXiv:1101.0438.

\bibitem {Gujrati-III}P.D. Gujrati, Phys. Rev. E \textbf{85}, 041129 (2012);
P.D. Gujrati, arXiv:1101.0431.

\bibitem {Einstein-BrownianMotion}A. Einstein, Ann. Phys. \textbf{17}, 549,
(1905), appearing in \textit{The Collected Papers of Albert Einstein}, English
translation by Anna Beck, Princeton U.P., Princeton, NJ. (1989), Vol. 2, pp. 123--134.

\bibitem {Langevin}P. Langevin, C. R. Acad. Sci. (Paris) \textbf{146}, 530
(1908); appears translated in English in D.S. Lemons and A. Gythiel, Am. J.
Phys. 65, 1079 (1997).

\bibitem {Chandrasekhar}S. Chandrasekhar, Rev. Mod. Phys. \textbf{21}, 383 (1949).

\bibitem {Sekimoto-Book}K. Sekimoto, \textit{Stochastic Energetics}, Springer,
Berlin (2010).

\bibitem {Marconi}U. M. B. Marconi, J. Chem. Phys. \textbf{124}, 164901 (2006)
and references therein.

\bibitem {Romanczuk}P. Romanczuk, M. B\"{a}r, W. Ebeling, B. Lindner, and L.
Schimansky-Geier, Eur. Phys. J. Spec. Top. \textbf{202}, 1 (2012).

\bibitem {Kapral-2017}P. Gaspard, and R. Kapral, J. Chem. Phys. \textbf{147},
211101 (2017).

\bibitem {Fodor}\'{E}. Fodor and M. C. Marchetti, Physica A, \textbf{504}, 106 (2018).

\bibitem {Beck}C. Beck, Prog. Theor. Phys. Supp., \textbf{162}, 29 (2006).

\bibitem {Granick}B. Wang, S.M.\ Anthony, S.C. Bae, and S. Granick, Proc.
Natl. Acad. Sci., \textbf{106}, 15160 (2009).

\bibitem {Zheng}C. Xue, X. Zheng, K. Chen, Y. Tian, and G.\ Hu, J. Phys. Chem.
Lett. \textbf{7}, 514 (2016).

\bibitem {Evans}M.G. McPhie, P.J. Daivis, I.K. Snook, J. Ennis, and D.J.
Evans, Physica A, \textbf{299}, 412 (2001).

\bibitem {Gaspard}P. Gaspard, Physica A, \textbf{552}, 121823 (2019).

\bibitem {Mizuno}D. Mizuno, C. Tardin, C.F. Schmidt, and F.C. MacKintosh,
Science, \textbf{315}, 370 (2007).

\bibitem {Hanggi}P. H\"{a}nggi and F. Marchesoni, Rev. Mod. Phys. \textbf{81},
387 (2009).

\bibitem {Huang}R. Huang, I. Chavez, K.M. Taute, B. Lukic, S. Jeney, M.G.
Raizen, and E. Florin, Nature Phys. \textbf{7}, 576 (2011).

\bibitem {Keizer}J. Keizer, \textit{Statistical Thermodynamics of
Nonequilibrium Processes}, Springer-Verlag, New York (1987).

\bibitem {Mazur}P. Mazur and D. Bedeaux, Physica A \textbf{173}, 155 (1991);
see also, Biophys. Chem. \textbf{41}, 41 (1991).

\bibitem {Pomeau}Y. Pomeau and J. Piasecki, C.R. Physique, \textbf{18}, 570 (2017).

\bibitem {Zwanzig}R. Zwanzig, \textit{Nonequilibrium Statistical Mechanics},
Oxford (2001). \ \ \ \ \ \ \ \ \ \ \ \ \ \ \ \ \ \ \ \ \ \ \ \ \ \ \ \ \ \ \ \ \ \ \ 

\bibitem {Evans-book}D.J. Evans and G.P. Morriss, \textit{Statistical
Mechanics of Nonequilibrium Liquids}, ANU\ Press, Canaberra (2007).

\bibitem {Gujrati-GeneralizedWork}P.D. Gujrati, arXiv:1702.00455.

\bibitem {Gujrati-GeneralizedWork-Expanded}P.D. Gujrati, arXiv:1803.09725; see
also P.D. Gujrati, Phys. Lett. A \textbf{384}, 126460 (2020). 

\bibitem {Kapral0}R.I. Cukier, R. Kapral, J.R. Lebenhaft, and J.R. Mehaffey,
J. Chem. Phys. \textbf{73}, 5244 (1980).

\bibitem {Reichl}L.E. Reichl, \textit{A Modern Course in Statistical Physics},
Second Edition, John Wiley and Sons, New York (1998).

\bibitem {Note-Notation}A change $d\varphi=d_{\text{e}}\varphi+d_{\text{i}%
}\varphi$ in some extensive macroquantity $\varphi$ is
\cite{Gujrati-I,Gujrati-II,Gujrati-GeneralizedWork,Gujrati-GeneralizedWork-Expanded}
partitioned into $d_{\text{e}}\varphi$, which is exchanged with the medium,
and $d_{\text{i}}\varphi$, which is internally generated within the system.
The second law controls the behavior of $d_{\text{i}}\varphi$. A similar
partition for a microquantity $d\varphi_{k}$ can also be carried out but
$d_{\text{i}}\varphi_{k}$ is not controlled by the second law.

\bibitem {Debenedetti}P.G. Debenedetti, \textit{Metastable Liquids; Conc epts
and Principles}, Priceton University Press, Princeton (1996).

\bibitem {Gujrati-Hierarchy}P.D. Gujrati, Entropy, \textbf{20}, 149 (2018);
see Sec. 8.1 and Eq. (58).

\bibitem {Gujrati-Entropy1}P.D. Gujrati, arXiv:1304.3768; see also
\cite{Gujrati-Entropy2}.

\bibitem {Gujrati-Entropy2}P.D. Gujrati, Entropy, \textbf{17}, 710 (2015).

\bibitem {Landau-Fluid}L.D. Landau and E.M. Lifshitz, Fluid Mechanics,
Pergamon Press, Oxford (1982).

\bibitem {Gujrati-Entropy-Note}The derivation of $p_{k}$ in Sec. 6.2 in Ref.
\cite{Gujrati-Entropy2} assumes fixed macrofields with corresponding extensive
microquantities fluctuating so that $p_{k}\propto\exp[-(E_{k}+PV_{k}%
-\mathbf{F}_{\text{BP}}\mathbf{\cdot R}_{k})/T]$; see Eq.
(\ref{IEQ-probabilities0}). We are also interested in the ensemble with fixed
extensive macroquantities but fluctuating microfields. The derivation is
easily extended to handle this situation with the result given in Eq.
(\ref{IEQ-probabilities}).

\bibitem {Gujrati-Fluctuations}P.D. Gujrati in \textit{Recent Research
Developments in Chemical Physics, }edited by S. Pandelai (Transworld Research
Network, Trivandrum, Kerala, India, 2003), Vol. 4, p. 243; arXiv:cond-mat/0308439.

\bibitem {Mishin}Y. Mishin, Annals of Physics \textbf{363,} 48 (2015).

\bibitem {Israelachvili}J.N. Israelachvili, \textit{Intermolecular and surface
forces}, Third edition, Academic Press, U.S.A. (2011).

\bibitem {Woodcock}L.V. Woodcock, Chem. Phys. Lett. \textbf{10,} 257 (1971).

\bibitem {Bussi}G. Bussi, D. Donadio, and M. Parrinello, J. Chem. Phys.
\textbf{126}, 014101 (2007).

\bibitem {Andersen}H. C. Andersen, J. Chem. Phys. \textbf{72}, 2384 (1980).

\bibitem {Nose}S. Nos\'{e}, J. Phys.: Condens. Matter, \textbf{2
(Supplement)}, SA115 (1990).

\bibitem {Frenkel}D. Frenkel and B. Smit, \textit{Understanding Molecular
Simulation}, Academic Press, San Diego (1996).

\bibitem {Gujrati-Heat-Work0}P.D. Gujrati, arXiv:1105.5549; see also
\cite{Gujrati-Entropy2}.

\bibitem {Gujrati-Heat-Work}P.D. Gujrati, arXiv:1206.0702; see also
\cite{Gujrati-Entropy2}.
\end{thebibliography}
\end{document}